\documentclass[sigconf,authorversion,nonacm, 10pt]{acmart}
\usepackage{ifthen}
\usepackage{xcolor}
\usepackage{blindtext}
\usepackage{algorithm}
\usepackage{graphicx}
\usepackage{amsmath}
\usepackage[noend]{algpseudocode}
\usepackage{subfig}
\usepackage{xspace}
\usepackage{amsthm}
\usepackage{balance}

\usepackage{graphicx}
\usepackage{footnote}
\usepackage{comment}
\usepackage{kantlipsum}
\makesavenoteenv{tabular}
\makesavenoteenv{table}
\usepackage{listings}
\usepackage{colortbl}
\usepackage[skins]{tcolorbox}
\usepackage[framemethod=tikz]{mdframed}

\makeatletter
\def\thanks#1{\protected@xdef\@thanks{\@thanks
        \protect\footnotetext{#1}}}
\makeatother

\newcommand{\name}{{DMart}\xspace}

\makeatletter
\let\OldStatex\Statex
\renewcommand{\Statex}[1][3]{%
  \setlength\@tempdima{\algorithmicindent}%
  \OldStatex\hskip\dimexpr#1\@tempdima\relax}
\makeatother
\algnewcommand{\LeftComment}[1]{\OldStatex \(\triangleright\) #1}
\algnewcommand{\LineComment}[1]{\OldStatex \(\triangleright\) #1}

\newcommand{\para}[1]{{\textbf{{#1}}}}

\algnewcommand{\IOComment}[1]{\OldStatex \(\triangleright\) #1}
\algnewcommand{\firstLeftComment}[1]{\OldStatex \(\indent\triangleright\) #1}
\algnewcommand{\secondLeftComment}[1]{\OldStatex \(\indent\indent\triangleright\) #1}
\algnewcommand{\thirdLeftComment}[1]{\OldStatex \(\indent\indent\indent\triangleright\) #1}
\algrenewcommand\algorithmicwhile{\textbf{With probability}}
\definecolor{LightCyan}{rgb}{0.88,1,1}
\definecolor{celadon}{rgb}{0.67, 0.88, 0.69}

\newcommand{\exclude}[1]{}
\newcommand{\showComments}{yes}
\newcommand{\note}[2]{
    \ifthenelse{\equal{\showComments}{yes}}{\textcolor{#1}{#2}}{}
}

\newcommand*\reqt[1]{%
\begin{tikzpicture}[baseline=(char.base)]
  \node[
    draw=black, 
    circle, 
    inner sep=1.2pt,   
    fill=white, 
    line width=1.1pt,  
    font=\sffamily\bfseries\small 
  ] (char) {#1};
\end{tikzpicture}}

\usepackage[english]{babel}
\usepackage{blindtext}
\usepackage{color}
\usepackage{graphicx} 
\usepackage{outlines}
\usepackage{xspace}
\usepackage[labelfont=bf, textfont=normalfont, font=small]{caption}

\usepackage[subtle]{savetrees}
\usepackage{enumitem}

\renewcommand\footnotetextcopyrightpermission[1]{} 
\setcopyright{none}

\acmDOI{}

\acmISBN{}

\acmPrice{}

\title{Bring Your Own Objective: Inter-operability of Network Objectives in Datacenters}
\author{Sanjoli Narang$^1$ \hspace{10pt} Anup Agarwal$^2$ \hspace{10pt} Venkat Arun$^3$ \hspace{10pt} Manya Ghobadi$^1$\\[10pt]}
\affiliation{%
  \institution{$^1$Massachusetts Institute of Technology \hspace{10pt} $^2$Google \hspace{10pt} $^3$UT Austin \\[25pt]}
  \country{}
}

\begin{document}

\setlength{\abovedisplayskip}{1pt}
\setlength{\belowdisplayskip}{2pt}

\begin{abstract}
Datacenter networks are currently locked in a ``tyranny of the single
objective''. While modern workloads demand diverse performance goals, ranging
from coflow completion times, per-flow fairness, short-flow latencies,
existing fabrics are typically hardcoded for a single metric. This rigid
coupling ensures peak performance when application and network objectives align,
but results in abysmal performance when they diverge. 

We propose \name, a decentralized scheduling framework that treats network
bandwidth as a competitive marketplace. In \name, applications independently
encode the urgency and importance of their network traffic into autonomous bids,
allowing diverse objectives to co-exist natively on the same fabric. To meet the
extreme scale and sub-microsecond requirements of modern datacenters, \name
implements distributed, per-link, per-RTT auctions, without relying on ILPs,
centralized schedulers, or complex priority queues.

We evaluate \name using packet-level simulations and compare it against network
schedulers designed for individual metrics, e.g., pFabric and Sincronia. \name
matches the performance of specialized schedulers on their own ``home turf''
while simultaneously optimizing secondary metrics. Compared to pFabric and Sincronia,  \name
reduces deadline misses by 2$\times$ and coflow completion times by 1.6$\times$ respectively, while matching pFabric short-flow completion times. \looseness=-1

\end{abstract}
\maketitle
\pagestyle{plain}
\vspace{-0.2cm}
\section{Introduction}
\label{sec:intro}

Over the years, networking research and practice have developed many notions of how bandwidth should be allocated among applications sharing a network. The most common default is to divide bandwidth equally among flows that share a bottleneck, with extensions such as proportional fairness or max-min fairness to handle multiple bottlenecks~\cite{vjacobson_cc, dctcp, max_min, num-kelly, fairness_analysis, fair_queueing}. 

It has also long been recognized, however, that such allocations—while seemingly fair—often fail to optimize application level objectives. As a result, a large body of work has explored alternative allocation strategies. Some systems explicitly prioritize or assign greater weight to certain flows~\cite{num-kelly}. Others bias allocations toward short flows, a strategy that reflects common application semantics and minimizes average flow completion time~\cite{SRPT_1966, SRPT_optimality_1968, pfabric, pias, homa, ndp}. Deadline-aware schedulers exploit knowledge of flow deadlines~\cite{d2tcp, d3, premtive_schdle, karuna} while coflow scheduling mechanisms consider groups of related flows together~\cite{varys, aalo, sincronia}. More recently, bandwidth allocation has been tailored to specific application domains, including machine learning training~\cite{cassini_hotnets, cassini_nsdi24, mltcp_arxiv, mltcp_hotnets} and video streaming~\cite{minerva}. Diversity exists not only in how fairness is defined, but also in how these policies are implemented. Some designs rely solely on end-host mechanisms, while others introduce in-network support~\cite{dctcp,pfabric,homa,ndp,xcp,rcp,dcqcn} or centralized schedulers~\cite{bwe,swan,fastpass,flexplane,sincronia}. Similarly, some systems assume fixed routing, whereas others dynamically pick routes in response to congestion~\cite{plb,hedera,presto,conga}. \looseness=-1

Despite this extensive design space, existing mechanisms typically target different—and often incompatible—points within it. Real networks host many applications with competing objectives, forcing operators to settle on a single common-denominator solution. Such solutions are rarely disastrous, but they are rarely optimal either. Most often, deployed networks continue to rely on flow-level fairness, occasionally augmented with coarse-grained priorities or weights~\cite{virtual_qos, qos_dngrade, hpc_qos, price_priority, rfc_qos}, and much of the richer research literature has seen limited deployment.

This paper proposes a new common denominator in which each application can independently choose (1) its objective and (2) degree of centralized control among its flows. This design is future proof towards diversity and innovation: applications may define objectives and mechanisms beyond those anticipated today, yet coexist as long as they satisfy a small set of minimally restrictive constraints. We focus on datacenter networks, which host diverse applications with heterogeneous performance requirements.

We take a novel approach to leverage markets to solve the problem of incomparability across objectives using money. Resources are assigned prices that reflect supply and demand, allowing participants with diverse and even arbitrarily complex objectives to interact through a simple interface: how much they are willing to pay. \looseness=-1

We import this principle into network resource allocation by designing \name. For example, consider a short flow that wants to quickly squeeze through a congested link versus a straggler flow in machine learning. Urgency is contextual: if the straggler is blocking hundreds of GPUs it is more urgent than the short flow. If not, then short flow can dominate and complete with minimal system impact. \name is a datacenter-wide link auction mechanism which resolves these tradeoffs through bidding at RTT timescales, rather than hard-coding a single global objective or implementation model. While details of the interface and currency model are discussed later, adopting a market-based abstraction raises four challenges. \looseness=-1

First, making scheduling decisions at RTT granularity would require datacenter-wide auction millions of times per second, which is infeasible at scale. \name uses switch-local auctions so switches make allocation decisions independently. When allocations require coordination across links, \name uses a decentralized mechanism that converges quickly without a global controller.

Second, some objectives may require the applications to adapt their bids to network conditions. For example, a deadline flow should bid less and yield to urgent traffic if congestion is transient. Otherwise, it should bid aggressively as slack shrinks.  
This involves accounting for how current decisions affect its future chances of being served. Such adaptations can lead to non-convergence and oscillations as participants react to each other. \name imposes minimal, well-defined restrictions on bidding behavior to ensure global convergence. \looseness=-1

Third, markets can incentivize participants to bid strategically to gain performance while hurting others, which adds unnecessary complexity. This concern is not merely about trust: even well-intentioned application developers may inadvertently induce strategic behaviour when their algorithms, built to bid adaptively, learn to bid strategically instead.  We therefore introduce mechanisms to keep bidding algorithms simple and promote truthful participation.

Finally, not all application developers want to participate directly in a market simply to send packets. \name supports third-party bidding libraries and simple service tiers (e.g., flat-price best-effort) to participate in auctions on an application’s behalf, hiding complexity and financial risk while preserving the common interface. This allows applications to opt into market-based allocation only when beneficial, without imposing an additional burden on others. \looseness=-1

Together, these ideas enable a network in which diverse objectives and mechanisms coexist through a shared market abstraction. Throughout this paper, we use \emph{money} as a convenient unit of account, but the system does not require real currency. Operators can implement bids using virtual credits/tokens. Our contributions are as follows : \vspace{0.1cm} \\ 
\reqt{1} We develop switch-local mechanisms to implement the auctions fabric efficiently and give a convergence argument under honest participation.\\
\reqt{2} We design bidding strategies for four objectives: (i) minimizing average flow completion time, (ii) maximizing deadline satisfaction, (iii) fair sharing, and (iv) minimizing average coflow completion time. The first three are decentralized, while the last uses application-local centralization.\\
\reqt{3} We demonstrate that all four objectives coexist within the same network and that \name dominates existing baselines by 2$\times$ in deadline miss rate, 1.6$\times$ in average coflow completion time while maintaining avergae FCT of short flows within 5\% of the most optimal datapath baseline across a wide range of mixed-workload mixed-objective simulations. 

This work does not raise any ethical issues.
\section{Motivation}
\label{sec:motivation}

In this section, we motivate the need for a scheduling primitive that can arbitrate among \emph{heterogeneous objectives} in modern datacenters, and explain why prior approaches fall short. \looseness=-1

\para{Heterogeneous Objectives.} 
Datacenters increasingly host workloads with heterogeneous and often conflicting performance objectives~\cite{fb_dc, aqua, gcloud}. User-facing services demand extremely low latency, while other workloads require sustained throughput or strict deadlines.
A large body of prior work has focused on optimizing individual objectives in isolation---for instance, Shortest Remaining Processing Time (SRPT)~\cite{pfabric, ndp, homa, pias} for minimizing average Flow Completion Time (FCT), Earliest Deadline First (EDF)~\cite{d3, d2tcp} for maximizing the number of deadlines met, and Smallest Effective Bottleneck First (SEBF)~\cite{varys, sincronia, aalo} for reducing average Coflow Completion Time (CCT), or resource interleaving~\cite{mltcp_hotnets, cassini_nsdi24, mltcp_arxiv, crux} for periodic ML traffic. However, these policies underperform when flows with different objectives contend for bandwidth, even on a single link. Figure~\ref{fig:mot_exp} illustrates how single-objective schedulers fail to satisfy heterogeneous SLOs simultaneously even in a simple case of 3 flows on a single link despite the existence of an optimal allocation that meets all SLOs (Figure~\ref{fig:mot_exp}(e)).

\para{Future-Proof Scheduling for Emerging Workloads.} Prior work on multi-objective scheduling is limited in scope, typically supporting only a narrow set of objectives: often no more than two~\cite{karuna, hedera, mocc_ieee, d2tcp}. 
For example, Karuna~\cite{karuna} combines priority queues with specialized congestion control to support mixtures of latency and deadline-sensitive traffic. However, such approaches do not generalize to broader or changing set of objectives often requiring redesign when new application classes introduce unfamiliar performance requirements. \looseness=-1 

\begin{figure}[t]
    \centering
    \includegraphics[width=0.48\textwidth]{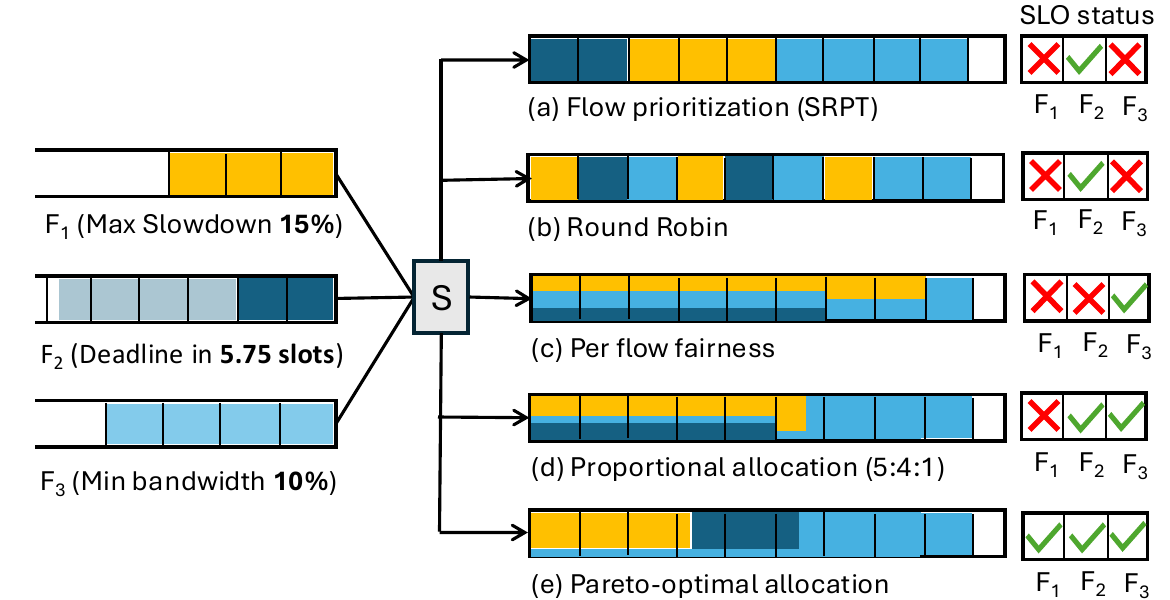}
    \caption{\textbf{Single-link example with heterogeneous SLOs} (F1 slowdown, F2 completion time, F3 minimum bandwidth). Outcomes: (a) SRPT (67\%, 2 units, 0\%); (b) Round Robin (233\%, 5 units, 0\%); (c) Fairness (267\%, 6 units, 33\%); (d) Proportional(5:4:1) (187\%, 5 units, 10\%); (e) SLO-optimal (11\%, 5.55 units, 10\%).}
    \label{fig:mot_exp}
\end{figure}

\para{Class-based Isolation.} Production datacenters mitigate conflicts by enforcing isolation through class-based QoS, assigning coarse-grained and slow-changing priorities to workload categories~\cite{virtual_qos, qos_dngrade, hpc_qos, price_priority, rfc_qos}. Within a class, the network is optimized for a single objective using packet scheduling or congestion-control~\cite{pias, homa, ndp, d3, d2tcp, dctcp, qclimb}. This approach is fundamentally limited because hardware restricts the number of classes and, more importantly, assigning classes itself is challenging because every team believes its traffic is critical. \looseness=-1

We argue that any scheduling decision depends on two distinct notions: \textit{entitlement} and \textit{urgency}. Entitlement captures the long-term privilege or importance of a flow, while urgency reflects its instantaneous need for service under its objective. For instance, a critical system task may have high entitlement, yet low urgency when its deadline is far away. Existing systems have been conflating these notions into a single priority value, obscuring their different roles. In practice, entitlement corresponds to policy or class, whereas urgency is objective-driven, dynamic, and increasingly diverse. \looseness=-1

\para{The Challenge: Expressiveness Gap.} When diverse workloads compete, allocation requires making trade-offs across qualitatively different performance objectives. Applications must negotiate with each other on how much they are willing to trade their own performance for others. In addition, applications must evaluate their urgency in the form of the opportunity cost of giving in : ``If I do not transmit now, can I recover later?''. We describe this challenge as the expressiveness gap: applications need a mutually intelligible way to express their priorities. 

For example, a deadline-driven flow with ample slack should yield, while a brief, latency-critical control message should be able to momentarily seize bandwidth.
Moreover, urgency is contextual: if many flows are stragglers, none can be treated urgent; if a single flow is delaying thousands of GPUs, it becomes extremely high priority. Static policies(e.g. SRPT, EDF etc.) only consider flow's local state and ignore the context, load and future opportunity cost (winning service is more expensive during high loads). In single-objective settings, every flow faces similar opportunity cost, hence urgency collapses to some static policy. In multi-objective settings, urgency is inherently dynamic and must account future opportunity costs. This is precisely what enables an Optimal solution in Figure~\ref{fig:mot_exp}(e) where flows cooperatively yield just enough service to satisfy all SLOs.

\para{Market Abstraction.}
Even on a single link, multi-objective optimization is NP-hard, making centralized solutions impractical at RTT timescales. 
We therefore reformulate network resource allocation as a market-clearing problem, rather than directly solving multi-objective optimization, which is both computationally intractable and ill-defined across heterogeneous objectives. Money (or tokens) provides a common scale for comparing urgency across workloads: flows express their urgency through a willingness to pay. We propose a deployable datacenter-wide \textbf{auction} mechanism repeated every RTT. Budgets encode \emph{entitlement}, while per-RTT bids encode dynamic \emph{urgency}. Contention is resolved by allocating capacity to the highest bidders and charging for service. Repeated auctions naturally solve the "Expressiveness Gap" by incorporating future opportunity cost in the bids. \looseness=-1

This abstraction is future-proof: new objectives can be expressed as new bidding policies (\S\ref{sec:rl}) without re-engineering the network. However, realizing such a market inside a packet network introduces engineering-economics trade-offs: operation at line rate, coordination across links, stability under flow churn, and robustness to strategic behavior. Section~\S\ref{sec:overview} provides an overview of the key mechanisms in \name that make this market solution successful in large networks while preserving simplicity and deployability.
\vspace{-0.2cm}
\section{Related Work} 
\para{Network Link Auctions.}
Prior work explored auction-based bandwidth allocation, notably Progressive Second Price (PSP) \cite{psp_for_network, psp2} and follow-ups~\cite{psp_1, psp_2, psp_3, psp_4}, as well as market-oriented programming~\cite{mop}. These approaches focus on static bandwidth division under fixed utility-bandwidth curves. In contrast, \name uses auctions as a \emph{scheduling primitive}, where bids encode \emph{application objectives} rather than static utility–bandwidth tradeoffs. \looseness=-1

\para{Centralized Bandwidth Allocation.} 
SWAN~\cite{swan} and BwE~\cite{bwe} represent a complementary line of centralized bandwidth allocation.
Both approaches rely on a centralized global optimizer, with control timescales of seconds or longer. In contrast, \name lets each application handle its own objective locally, flows adapt quickly (RTT scales) without relying on a global controller. BwE’s per-flow utility model cannot easily capture flow-coupled objectives such as coflows without changing the core design. \name supports centralized control at any desired scope: from individual flows, to application-level controllers, to datacenter-wide coordination. 

\para{Cloud Network Sharing.} FairCloud~\cite{faircloud} studies how to share network among tenants in a cloud environment by navigating trade-offs between required properties (min throughput guarantee, fairness etc.). FairCloud provides a few static allocation policies at different points on the tradeoff curve. Unlike static allocation, \name being an RTT-scale scheduler navigates a time varying trade-off between application objectives through auctions.

\section{Design Overview}
\label{sec:overview}
This section describes the \emph{mechanism} of \name: the switch-local per-link auctions, the participation interface, and the additional rules needed to compose per-link decisions into a work-conserving datacenter fabric. Section~\S\ref{sec:theory} addresses bidder's incentives and stability dynamics of \name. Section~\S\ref{sec:implementation} dives into implementation details. \looseness=-1
\vspace{-0.2cm}
\subsection{Auction-based Link Allocation}
\name is a datacenter-wide fabric that resolves contention by auctioning link service to competing flows. Switches run auctions at egress links and maintain a ranked list of top-k winners per link. Each link admits the number of winners enough to saturate the link capacity so $k$ varies with the demands of the competing flows. 
Because end-to-end progress requires all links on a flow’s path, a flow submits the same bid to every link it traverses, consistent with prior multi-link auction models~\cite{psp_for_network, psp2}. A flow’s bid represents its willingness to pay for service, providing a common unit of comparison across heterogeneous objectives. The auctions are \textit{second-price}; the winning flow is charged the highest losing bid for service. Second-price payments discourage applications from always claiming the highest priority and force winners to consider the externalities they impose on others. Switches only signal eligibility to send; rate control remains at end hosts. \looseness=-1

Figure~\ref{fig:topo_auction} illustrates the auction process on a leaf-spine topology (all links have same capacity, $k=1$). Flow-1 and Flow-2, both capable of saturating the link, compete at a downstream switch, where Flow-1 wins with a bid of 30\$ per link and pays a clearing price of 10\$ which is equal to the bid of loser Flow-2. The total amount Flow-1 pays is equal to the sum of clearing prices on all links (=10\$). 
\begin{figure}[t]
    \centering
    \includegraphics[width=0.4\textwidth]{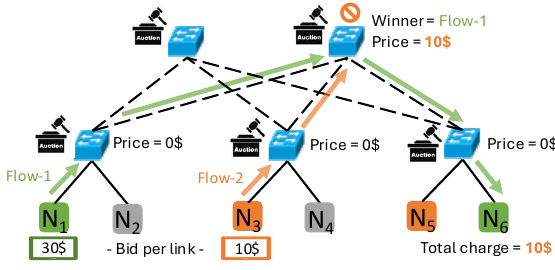}
    \caption{Illustration of end-to-end flow control using independent per-link second-price auctions on a 2-layer topology.}
    \label{fig:topo_auction}
\end{figure}

As discussed in Section~\S\ref{sec:motivation}, the auction bids encode urgency of flows. A single bid is akin to a static priority. However, encoding urgency inherently requires time-varying bids. To do this, auctions repeat every RTT at the switches entitling the winners to exclusive link service only for one RTT. Each RTT, the sender registers a fresh bid on all switches on its path using a lightweight probe. The sender transmits data at line rate for one RTT only when the flow wins \emph{all} auctions along its path; otherwise it pauses and continues probing. A flow $F_i$ with total size $D_i$ bytes traversing a path of fixed links, each with capacity $\mathcal{C}$ bytes per second must win approximately $N_i = D_i/(\mathcal{C} \times RTT)$ auctions over time. \looseness=-1

Flows are free to generate bids and update bidding policies over time, using fixed heuristics to adaptive or learning-based algorithms, implemented centrally or at endpoints. In practice, adaptive algorithms that respond to evolving market conditions (e.g., workload drift or new objectives) and incorporate precise future market conditions in bids best represent dynamic urgency. 
\name provides a library of adaptive bidding agents for a bunch of common objectives and an API to develop new ones (Section~\S\ref{sec:rl}).  
However, our analysis of the market dynamics and convergence does not mandate learning and do not restrict the type of update algorithms. It only requires a few regularity conditions on bidding policies (Section~\ref{sec:conv_theory}) that we trust all bidders to follow diligently. Innovative bidders may employ learning-based updates while uninterested wlos/application may stick to a constant best effort bidding.

\vspace{-0.2cm}
\subsection{Challenges in Market Design}

Independent per-link, per-RTT auctions are simple and react quickly to model time-varying urgency, but raise several challenges in multi-hop fabrics, especially when bidders adapt over time. \name addresses these challenges using switch-local mechanisms and a constrained bidder interface.

\para{Challenge-1 : Work Conserving Fabric.}  
As a flow traverses multiple links along its path, it must acquire service on all of them to make progress. Because link auctions are conducted independently, a flow may win an upstream auction yet lose at a downstream link, thereby reserving the first link without sending any data. Such uncoordinated decisions can lead to deadlock and wastage of link capacity despite outstanding demand. Figure~\ref{fig:topo_auction_deadlock} shows such a situation using (all links have same capacity). Flow-1 blocks Flow-2 at $S_1$ while itself being blocked downstream by Flow-3 at $S_5$, allowing only Flow-3 to make progress, even though Flow-2 could have transmitted as well.

\para{Solution : Overcommitment for Coordination. }
We introduce \textit{overcommitment}, where a switch that detects downstream blocking temporarily increases the winner quota $k$ to admit additional flows. Overcommitment does not increase link capacity; it only prevents upstream idling when a nominated winner cannot transmit. This is enabled by per-flow feedback to switches, where packets indicate whether the flow made progress in the previous RTT. Referring back to Figure~\ref{fig:topo_auction_deadlock}, after Flow-1 reports no progress in the previous RTT, switch $S_1$ temporarily increases the winner quota to admit Flow-2, breaking the deadlock. The winning quota restores back to 1 once Flow-1 is unblocked by Flow-3 at switch $S_5$. \looseness=-1

\begin{figure}[t]
    \centering
    \includegraphics[width=0.4\textwidth]{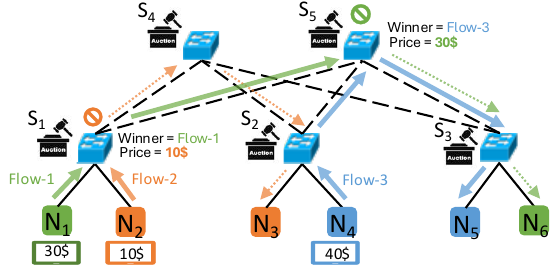}
    \caption{\textbf{Link idling and deadlock in link auctions} due to dependencies. Overcommitment increases winners quota at $S_1$ to unblock.}
    \label{fig:topo_auction_deadlock}
\end{figure}

\para{Challenge-2 : Translating Objectives to Bids. }  Flow objectives are defined end-to-end (e.g., minimizing completion time, meeting deadlines etc.), yet the switch auctions expect a simple per-RTT bid reflecting its urgency. Each flow must therefore translate its objective into bids under its budget.

\para{Solution: Trusted Bidding API.} \name provides objective-specific \textit{bidding agents} that translate end-to-end objectives into per-RTT bids. Flows delegate their budget and bidding decisions to the agents after specifying their objective function over measurable state (e.g., dollars as a function of completion time or throughput). Each RTT, the agent maps the flow state into its bid. We describe \name's bidding-agent algorithms in Section~\S\ref{sec:rl}. Applications are also free to stick to a best-effort bidding policy or innovate their own bidding policies following \name's regularity conditions.

\para{Challenge-3 : Manipulative or Strategic Bidding. }
In repeated auctions, bids can affect future congestion and hence future prices, creating incentives to manipulate the market rather than express urgency. We refer to this inter-temporal coupling as \emph{endogeneity}. For example, a low-urgency flow might overbid early to delay competitors with strict deadlines, hoping that future prices fall once urgent flows miss deadlines and drop out unfinished. 

\para{Solution : Restricted Feedback and Deferred Charge.}

\name limits bidders' ability to observe and exploit this "bid to future price" causality (endogeneity). Instead of exposing instantaneous link-local clearing prices to learning agents, bidders observe only a coarse, global price distribution $\mathcal{F}$ that summarizes prevailing market conditions. In addition, payments are deferred and charged in aggregate upon flow completion rather than per auction round. With these settings, attempts to game future prices are more likely to backfire than to help. Section~\S\ref{sec:trust_auctions} provides more details and theory.

\para{Challenge-4 : Stability of Learning-driven Market.}
With adaptive bidding agents, repeated auctions induce a coupled learning-and-bidding dynamic: bidders adapt their bidding policies based on the global price signature which itself evolves in response to their bids. It can lead to oscillations, price-chasing and unstable allocation, especially when many agents update their bidding policies simultaneously. Without a stationary regime, prices never settle, bids cease to reliably reflect urgency, and performance becomes unpredictable. \looseness=-1

\para{Solution : Lipschitz Policy and Timescale Separation.}
To ensure stability, \name adopts a \textit{two-timescale} approach: bidding policies update more slowly than per-RTT auctions but faster than the datacenter workload drift. Formally, if $\tau_B$ is the bidding policy update period and $\tau_D$ the drift timescale, \name enforces :
\begin{equation}
    RTT << \tau_B << \tau_D
\end{equation}
In addition, \name restricts bidding agents to adopt strategies that are \textit{Lipschitz} continuous in the global price distribution $\mathcal{F}$. This constraint limits the sensitivity of bids to fluctuations in market conditions, preventing abrupt policy changes that could destabilize prices. 
Section~\S\ref{sec:conv_theory} provides the formal convergence discussion.

\vspace{-0.1cm}
\section{Design Rationale and Theory}
\label{sec:theory}
This section explains why the \name design in \S\ref{sec:overview} is a reasonable foundation for multi-objective datacenter scheduling. We address two questions; (i)\emph{incentives}: if we let applications bid for service repeatedly, do they have an incentive to report meaningful urgency, or to manipulate prices to hurt others? and, (ii) \emph{stability}: as bids adapt to changing load and objectives, do repeated auctions converge, or do they oscillate?
Section~\S\ref{sec:trust_auctions} analyzes incentive properties of repeated second-price auctions and shows how \name’s restricted feedback and charging model discourages gaming. Section~\S\ref{sec:rl} describes \name's bidding API to compute objective-specific bids using the global price distribution $\mathcal{F}$. Section~\S\ref{sec:conv_theory} analyzes convergence under adaptive bidding. \looseness=-1

\vspace{-0.2cm}
\subsection{Incentives in Second-Price Auctions}
\label{sec:trust_auctions}

Under \name's abstraction, each flow acts as a bidder repeatedly competing for transmission opportunities until it finishes. In each auction round, a bidder must decide how much urgency to report through its bid. There are two concerns that need to be addressed : (i) Urgency computation and (ii) Incentives to bid only that correct urgency without manipulation. We refer to the second condition as "truthful" bidding.  \looseness=-1

\para{Single-Shot Second-Price Auctions.}
Consider a sealed-bid second-price auction for a single indivisible good. Bidder $i$ submits a bid $b_i$ and has private valuation of the good $V_i$. Let $P_{-i} = \max_{j\neq i} b_j$ denote the highest competing bid. Bidder $i$ wins if $b_i > P_{-i}$ and pays $P_{-i}$, yielding utility $U_i$
\begin{equation}
U_i = (V_i - P_{-i}) \cdot \mathbb{I}\{b_i > P_{-i}\}.
\end{equation}
Under standard independent private-value assumptions, truthful bidding of valuation ($b_i=V_i$) is a weakly dominant strategy i.e. no bid value other than $V_i$ can further improve bidder $i$'s utility. We build on this property to reason about incentives in repeated auctions.

\para{Repeated Auctions and Marginal Valuations.}
In repeated auctions, the notion of a bidder’s true “valuation” changes compared to single-shot auction.  Winning or losing the current round 
changes the flow’s future payoffs.  
We capture this forward-looking payoff as the flow’s \emph{expected continuation utility} $U_i(\cdot)$: the total performance benefit expected from now until completion \emph{minus} the expected future prices paid along the way. The correct per-round bid is the bidder’s \emph{marginal continuation utility} $\Delta U_i(\cdot)$—the improvement in expected net utility from winning the current round rather than losing, plus any immediate per-round gain or loss. \looseness=-1

Let $S_t$ denote bidder $i$'s internal state at time $t$ (e.g., remaining size, slack, coflow context). Let $r_l(S_t)$ and $r_w(S_t)$ denote the immediate gain in the current auction round $t$ if the bidder wins and loses respectively. Then the bidder’s truthful urgency is given by its \emph{marginal continuation utility}:
 \begin{equation}
\label{eq:marg_util}
\begin{split}
& \Delta U_i(S_t) = U_i(S_{t+1} \mid S_t, \text{win}) - U_i(S_{t+1} \mid S_t, \text{lose}) \\
& \quad + r_w(S_t) - r_l(S_t)
\end{split}
\end{equation}
Thus, truthful bidding in repeated auctions corresponds to bidding $b_i(S_t)=\Delta U_i(S_t)$ every round.  

Winning or losing the current round affects (i) flow's future states and (ii) other bidders' future states and hence future prices. The second condition creates incentives for bidders to game the system. To eliminate this incentive, 

a \emph{price-taking} abstraction is required where an individual flow’s current bid $b_i(S_t)$ does not affect the distribution of future clearing prices $P_{t'}$ in any future round $t'>t$. Formally,
\begin{equation}
\label{eq:exo}
P_{t'} \;\perp\; b_i(S_t), \qquad \forall\, t' > t.
\end{equation}
The following theorem formalizes this insight : 
\begin{theorem}[Truthful Bidding under Price-Taking Environment]
If the market is price-taking, as in Equation~\ref{eq:exo}, then bidding truthful marginal utility $b_i(S_t)=\Delta U_i(S_t)$ in each round is a weakly dominant strategy.
\end{theorem}

\begin{proof}
Appendix~\S\ref{sec:app_ic} provides the full proof.
\end{proof}
\vspace{-0.1cm}

\para{Price-Taking in \name.} Approximate price-taking behaviour emerges when flows are short-lived, arrivals are highly dynamic, and any single bidder has negligible influence on aggregate prices. However, real datacenters do not always follow the price-taking abstraction. Therefore, \name \textit{induces} price-taking by limiting a bidder's ability to infer and exploit the causal impact of its own bids on future prices. For this, bidding agents may learn only from an aggregated \emph{global price distribution} $\mathcal{F}$, rather than link-local per-RTT clearing prices. Payments are deferred and charged in aggregate upon flow completion, reducing per-round strategic feedback. \looseness=-1

Truthful marginal bidding simplifies agent behavior by eliminating incentives to search for complex strategic deviations which in turn mitigates price-chasing and helps stabilize market dynamics. While \name considers future costs, it does not claim global optimality. Its main goal is a stable, robust decentralized mechanism that performs well across heterogeneous objectives.

\vspace{-0.3cm}
\subsection{Objectives and Bidding Agents}
\label{sec:rl} 
A best-effort bidder with no objective simply bids a fixed per-round valuation. However, most datacenter flows have end-to-end objectives (e.g., minimizing completion time, meeting deadlines, or achieving throughput) and their urgency depends on evolving state. Under \name's restrictions to enforce price-taking, we describe one algorithmic approach for translating Markovian objectives (objectives whose value depends only on current state) into per-round bids using global price distribution $\mathcal{F}$. We model flow's state transitions over repeated auctions assuming price per auction is sampled from $\mathcal{F}$ and estimate marginal continuation utility.

\para{From objective to per-round rewards}
An end-to-end objective can be expressed as an objective-level dollar utility function (say $u(o)$) as a function of the achieved objective value $o(.)$ (e.g. 50\$ for 15\% slowdown, 40\$ for 20\% slowdown). As service progresses, the flow’s state $S_t$ (remaining size, slack, etc.) determines both $o(S_t)$ and the sensitivity of $u(o)$ given by ($\frac{du(o)}{do}$) to immediate service (e.g. slowdown already 15\%, so $u(o) \leq$ 50\$). The sensitivity defines the per-round rewards $r_w(S_t)$ and $r_l(S_t)$, capturing the immediate gain or loss from service. We now need to compute \emph{marginal continuation utility} using this model setting (Equation~\ref{eq:marg_util}). 

\para{Computing bidding policies.}
We model per-flow bidding as a sequential decision problem over repeated auctions. Consider a flow $i$ in state $S_t$ at round $t$ submitting a bid $b\in\mathcal{B}$, and let $P_t \sim \mathcal{F}$ denote the second price in that round (the highest competing bid). Let $U(S_t)$ be the bidder’s expected future \emph{net utility} from $S_t$ onwards. Then $U(\cdot)$ satisfies the following state-transition equation: \looseness=-1
\begin{equation}
\label{eq:bell}
\begin{split}
U(S_t) = \mathbb{E}_{P_t\sim\mathcal{F}}
\Big[ &\mathbb{I}\{P_t<b\}\big(r_w(S_t)-P_t+U(S_{t+1}\mid S_t, win)\big) \\
& + \mathbb{I}\{P_t\ge b\}\big(r_l(S_t)+U(S_{t+1}\mid S_t, lose)\big) \Big]
\end{split}
\end{equation}
This formulation of state transition applies at any centralized scope— individual flows or groups of flows with a joint state. As per \S\ref{sec:trust_auctions}, bidding $b=\Delta U(S_t)$ constitutes the truthful urgency to bid mapping.
\begin{equation}
\label{eq:marg_util_policy}
\begin{split}
& b = \Delta U(S_t) =
\big(r_w(S_t)+\mathbb{E}[U(S_{t+1}\mid S_t, win)]\big) \\
& \quad - \big(r_l(S_t)+\mathbb{E}[U(S_{t+1}\mid S_t, lose)]\big).   
\end{split}
\end{equation}
 
Equation~\ref{eq:marg_util_policy} and Equation~\ref{eq:bell} are solvable by substitution and iterative \textit{dynamic programming}. Although this appears complex, the state space is small for common objectives (e.g., remaining RTTs or deadline slack), making dynamic programming tractable and often admitting closed-form solutions.
An application’s \emph{budget} enters through the scale of the objective-level utility function $u(o)$ and naturally gets absorbed in bid computation. \looseness=-1

\para{Multi-link Modelling.} $\mathcal{F}$ denotes the global end-to-end price distribution (the maximum bid along a path). A flow advances when its bid exceeds the maximum link price along its path. With overcommitment, this reduces path scheduling to a single-link auction at the most congested point making the modelling similar to a repeated single-link auction. \looseness=-1

\para{Role of Learning Agents.} Dynamic programming suffices to compute marginal continuation utiliy in our model. Due to induced price-taking behaviour, learning algorithms like RL does not threaten truthfulness and may even help in expanding model, such as forecasting non-stationary prices or handling objectives that are difficult to specify analytically. We nevertheless assume cooperative bidders and do not defend fully defend against large-scale adversarial behaviors such as coalitions.
\vspace{-0.2cm}
\subsection{Common Datacenter Objectives}
This section instantiates \name’s bidding framework for common datacenter objectives.
For each objective, we specify (i) the flow state $S_t$, (ii) a representative objective-level
utility function $u(o)$ chosen by the application, and (iii) the induced per-round win and loss rewards $(r_w, r_l)$. Substituting these into the state transition equation
(Equation~\ref{eq:bell}) yields objective-appropriate bidding policies. The utility shape is application-defined and reflects how users trade money against performance; \name does not fix or assume a canonical form. Our goal is not to derive new schedulers, but to show that widely-used objectives can be expressed through a
unified bidding interface. Table~\ref{tab:objectives} summarizes objective utilities. \looseness=-1

\begin{table*}[t]
\centering
\scriptsize
\setlength{\tabcolsep}{5pt}
\renewcommand{\arraystretch}{1.1}
\begin{tabular}{|
p{1.5cm} |
p{3.5cm} |
p{2.7cm} |
p{3cm} |
p{2.2cm} |
p{2.2cm} |
}
\hline
\textbf{Objective} 
& \textbf{State $S_t$} 
& \textbf{Objective value $o(S_t)$} 
& \textbf{Utility $u(o)$} 
& \textbf{Win reward $r_w(S_t)$} 
& \textbf{Loss reward $r_l(S_t)$} \\
\hline

\textbf{Average FCT} 
&  $(F)$ - Remaining time  
& $o(S_t) = F$
& $C - w\!\left(o - \frac{o^2}{2T}\right)$ 
& $0$ 
& $-w\!\left(1 - \frac{F}{T}\right)$ \\

\textbf{Deadline} 
& $(F,D)$ - (Remaining time, Slack)
& $o(S_t) = \mathbb{I}\{ F=0, D>0\}$
& $C \times \mathbb{I}\{o \}$ 
& $0$ or $C$
& $0$ \\

\textbf{Fairness} 
& (X) - EWMA Average throughput
& $o(S_t) = X$
& $w\log(o)$ 
& $\frac{w}{S_t}\,\alpha(1-X)$ 
& $-w\alpha$ \\

\textbf{Average CCT} 
& $(B_F)$ - Bottleneck Remaining time 
& $o(S_t) = B_F$
& $C - w\!\left(o - \frac{o^2}{2T}\right)$  
& 0 
& $-w\!\left(1 - \frac{F}{T}\right)$ \\

\textbf{Best Effort} 
& None 
& None
& Constant 
& 0 
& 0 \\

\hline
\end{tabular}
\caption{\textbf{Common datacenter objectives and their sample utility functions and per-round rewards for DP formulation}}
\vspace{-0.4cm}
\label{tab:objectives}
\end{table*}

\begin{figure}
    \centering
    \includegraphics[width=0.45\textwidth]{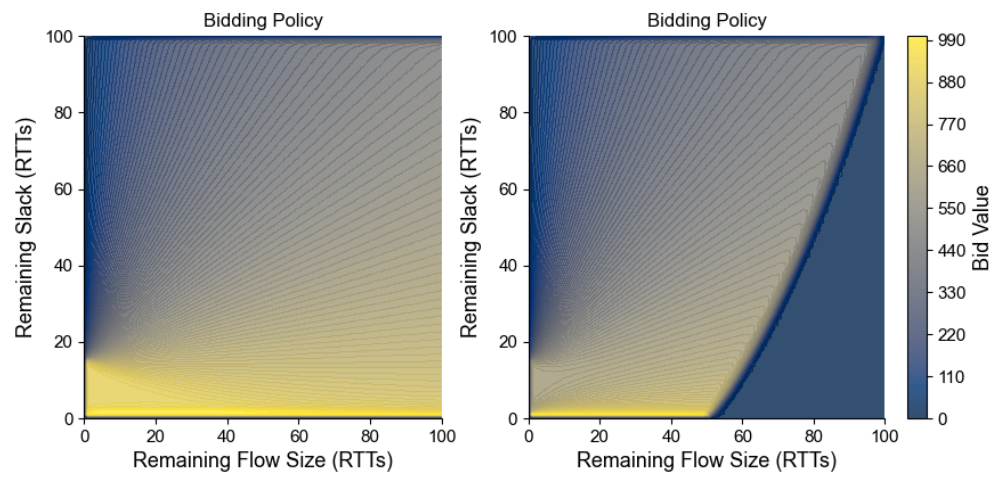}
    \caption{\textbf{Deadline-objective bidding as a function of state.} With a higher budget ($C=4000$), bidders remain active over a larger region of the state space than in the lower-budget case ($C=2000$).}
    \label{fig:dl_htmap}
\end{figure}

\noindent\textbf{Flow Completion Time (SRPT-like).}
The state $S_t$ is the remaining service time $F$ in RTTs. We use a concave, decreasing
quadratic utility in $F$, so delaying short flows incurs higher penalty than delaying long
flows. This causes short flows to bid more aggressively and recovers SRPT-style
prioritization when substituted into Equation~\ref{eq:bell}. The specific utility shape (e.g., quadratic vs.\ inverse) does not affect inter-flow ordering, only the relative scaling across objectives.

\noindent\textbf{Deadline met.}
The state $S_t$ is $(F,D)$, representing remaining service and deadline slack in RTTs. Utility is a step function: a fixed reward $C$ is obtained only if the flow completes before its deadline. 
Figure~\ref{fig:dl_htmap} depicts bidding policies obtained from solving state Equation~\ref{eq:bell} for two different values of budgets $C$. As slack decreases, bids grow high. Bidders on low budgets have a smaller space of active bidding. \looseness=-1

\noindent\textbf{Fairness (Proportional-Fairness Proxy).}
The state $S_t$ is the EWMA (factor-$\alpha$) average throughput over RTTs. Using logarithmic utility ensures flows with low throughput experience higher urgency for additional service. Bids increase sharply as a flow approaches starvation, while occasionally bidding high (to maintain throughput) otherwise.

\noindent\textbf{Coflow Completion Time (CCT).}
Bidding agents are coflow-centralized with visibility on all flows of the coflow. The state $S_t$ is the remaining service time on the bottleneck port. SEBF-like policy is implemented by reusing the SRPT-like utility on the bottleneck, while Varys is emulated by assigning non-bottleneck flows a deadline equal to the bottleneck completion time, reproducing MADD-style allocation~\cite{varys}. 

\noindent\textbf{Best Effort (Fixed Service Tier).}
Flows without explicit objectives submit a constant per-RTT bid without adapting to market or congestion.

Detailed derivations of corresponding bidding policies is discussed in Appendix~\S\ref{sec:app_api}.

\subsection{Convergence under Learning Agents}
\label{sec:conv_theory}

\name admits bidding agents that adapt their bidding policies over time based on observed market conditions. Adaptation operates on a slower timescale than per-RTT auctions but faster than datacenter drift. 

\para{Stationary Point. } Let $\mathcal{F}$ denote the global price distribution collected by \name over an update interval. 
Given $\mathcal{F}$, each agent $i$ computes a (possibly stochastic) bidding policy $b_i(\cdot)$ that maps its local state $S_t$ (e.g., remaining size or slack) to a bid.
Running these policies through per-link second-price auctions, together with flow arrivals and departures, induces a new realized price distribution $\mathcal{F}'$ in the next interval. Agents then update their policies based on $\mathcal{F}'$.  We represent this iterative market-level evolution through an operator $Mar(\cdot)$, defined as the composition of bidders’ responses and the resulting network dynamics:
\begin{equation}
\label{eq:fixed_point}
    \mathcal{F}' = \mathcal{H}(\pi(\mathcal{F})) \;=\; Mar(\mathcal{F}),
\end{equation}
where $\pi(\mathcal{F})=\{b_i(\mathcal{F})\}_i$ denotes the collection of bidding policies and $\mathcal{H}(\cdot)$ captures the induced auction outcomes, congestion, and flow churn. A steady-state market corresponds to a fixed point of this operator:
\begin{equation}
\label{eq:fp}
    \mathcal{F}^* = Mar(\mathcal{F}^*),
\end{equation}
at which the global price signature stabilizes under learning.

\para{Convergence Argument.}
Because bidding agents adapt in response to $\mathcal{F}$ while $\mathcal{F}$ itself is generated by their bids, naive learning can in principle lead to oscillations or price chasing if agents update too aggressively or synchronously. In repeated auctions, flow state trajectories under different bidding policies may diverge locally, since win/lose outcomes affect future states and bids. Since the datacenter market remains endogenous, we do not claim convergence for arbitrary learning dynamics. Instead, our argument relies on two assumptions of datacenter workloads:\\
(i) \emph{Bounded Divergence} of flow trajectories due to short flows. \\
(ii) \emph{Strong Mixing} induced by continuous flow arrivals.\\
These prevent perturbations due to a single flow from persisting indefinitely at the market level.

A single auction-clearing step in \name is inherently \emph{non-expansive}: clearing prices in a second-price auction are order statistics (e.g., the maximum losing bid), so small perturbations in bids are not amplified by the auction rule.
Under the above mixing assumptions, the price distribution generated over a flow’s lifetime exhibits bounded sensitivity to bidding policies.
Formally, we model this as a Lipschitz property:
\begin{equation}
    \mathcal{W}\!\left(\mathcal{H}(\pi), \mathcal{H}(\pi')\right)
    \;\le\;
    L_H \cdot \|\pi - \pi'\|,
\end{equation}
$\mathcal{W}(\cdot,\cdot)$ is Wasserstein distance between price distributions.

In addition, \name constrains bidding policies to vary Lipschitz smoothly with the observed global price distribution:
\begin{equation}
    \|\pi(\mathcal{F}) - \pi(\mathcal{F}')\|
    \;\le\;
    L_F \cdot \mathcal{W}(\mathcal{F},\mathcal{F}').
\end{equation}

Combining the two inequalities yields an effective bounded-sensitivity property of the induced market iteration:
\begin{equation}
    \mathcal{W}(\mathcal{H}(\mathcal{F}), \mathcal{H}(\mathcal{F}'))
    \;\le\;
    (L_H \cdot L_F)\, \mathcal{W}(\mathcal{F}, \mathcal{F}'),
\end{equation}

suggesting that the market operator behaves as a contraction in practice when $L_H \cdot L_F < 1$. While we do not claim a formal contraction guarantee for all workloads, our experiments consistently exhibit non-expansive behavior and rapid convergence to a stable price regime (Appendix~\S\ref{sec:app_conv}).

\para{Additional Measures.}
\name further enforces stability through explicit timescale separation by rate-limiting policy updates ($\tau_B \gg RTT$) and smoothing the price distribution (e.g., EWMA or KDE) before it is fed to learning agents.

\section{Implementation Deep Dive}
\label{sec:implementation}
We implement \name on top of the TCP stack, primarily at the socket interface. The design applies to other network stacks (e.g., RDMA) as long as the prototype has the following four components. 
(i) a compact market header carried in packets to convey flow identity, bids, and auction state;  
(ii) switch pipeline logic that parses this header and maintains per-link auction state;  
(iii) sender-side rate control that gates transmission based on auction outcomes; and  
(iv) a bidding-agent shim that maps flow state to per-RTT bids.

\subsection{Packet Metadata for Bidding}

Since auctions execute at switches, packets must convey enough information to (i) identify the flow ID, (ii) register its current bid, at the switches and (iii) return the end-to-end auction outcome to the sender. \name uses lightweight \textit{probe} packets for bid registration. Before sending data, a sender transmits one probe per RTT containing only a TCP header and a small market header (no payload). Probes are sparse and small, so their overhead is negligible and they do not materially contribute to congestion. Restricting bid registration to probes avoids distortions from data burstiness, loss recovery, or sender pacing. However, to avoid an RTT delay at headstart, we also allow SYN connection packets to work like probes too and help in bid registration.

\para{Market Header.} Each probe packet carries a 12-byte \textit{market header} alongside the TCP header (Figure~\ref{fig:pkt_hdr}). The header encodes a flow ID, application ID, bid value, price telemetry, and a compact flags field, in addition to a reserved byte for length. Switches parse the header, update auction flags, and forward the packet; the receiver echoes the header in ACKs, allowing the sender to learn the end-to-end outcome through ACKs. The flags field uses four bits:\\
(i) \textit{Auction status}: set by the sender; cleared by switch where the flow loses; once cleared, remains cleared downstream. \looseness=-1  \\
(ii) \textit{Packet type}: distinguishes probes from data packets.  \\
(iii) \textit{Previous end-to-end status}: indicates whether the flow made progress in the last RTT, inferred from the last ACK. \\ (iv) \textit{Bypass}: skips auction logic for ultra-short flows that complete within a single RTT.

\para{Price Telemetry.} 
To construct the global price distribution $\mathcal{F}$ used by bidding agents, each probe packet carries an additional field that tracks the \emph{maximum bottleneck price} along its path. If the flow is a winner, it records the highest bid from losers set, and if it is a loser, it records the lowest bid from the winners set. This field is updated at each switch as the maximum observed so far. The observed telemetry is logged across many RTTs on the senders and later globally aggregated offline (or by a control-plane service) to generate $\mathcal{F}$. Crucially, it is \textit{not} exposed to bidding agents to maintain Price-Taking abstraction.

\begin{figure}
    \centering
    \includegraphics[width=0.43\textwidth]{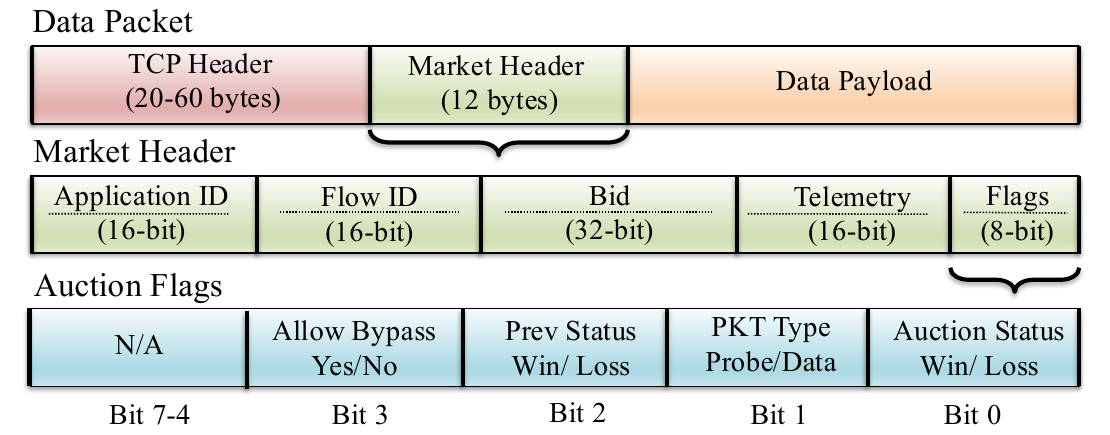}
    \caption{\textbf{TCP Packet with the new Market Header}}
    \vspace{-0.1cm}
    \label{fig:pkt_hdr}
\end{figure}

\subsection{Rate Control and Queuing}
\vspace{-0.2cm}
This subsection discusses the underlying transport and datapth of \name.

\para{Winners' Quota.} Each egress link maintains an overcommitment factor $k$--the number of simultaneous ``winners'' admitted for that link in the next RTT. Intuitively, $k{=}1$ corresponds to a strict single-winner auction; $k{>}1$ admits multiple winners to avoid upstream idling when the nominated winner cannot make end-to-end progress or when a single winner is unable to saturate link.
In our prototype, the default $k$ for a switch is set to be ratio between its egress-link capacity and the host access-link capacity, $C_{egress}/C_{host}$. If the host link is 50Gbps, while spine links are 200 Gbps as in Figure~\ref{fig:ls_topology}, then the spine switches have default overcommitment of 4.

\para{Per-link auction state (current vs. next RTT).} To implement per-RTT auctions, each switch maintains two auction bid records of size $k$ and a timer per egress link:
\\
(i) \textbf{Current winners.} The top-$k$ winning bids and flow IDs that are entitled to send during the current RTT. This record remains fixed during the current RTT.
\\
(ii) \textbf{Next winners.} The top-$k$ bids and corresponding flow IDs observed so far in the current RTT : potential winners in the next RTT. This record updates upon every probe packet until an RTT timer expires when the switch atomically promotes ``next winners'' to ``current winners'' and begins collecting bids for the following RTT.

\para{Adaptive Overcommitment.}
Although link is an indivisible good in \name, overcommitment temporarily admits multiple winners for work conservation. Switches use the \textit{Previous end-to-end status} flag to detect downstream blocking. If a flow is a current winner on this link but the previous end-to-end status indicates it did not make progress in the last RTT, the switch temporarily increases $k$ (e.g., by 1) to admit additional winners. Once the flow resumes end-to-end progress, the switch reduces $k$ back toward its default. 

\para{Startup and Pipeline Delays.}
As discussed before, to avoid an extra RTT of delay on start, we allow TCP SYN packets to participate in bid registration like probes. After startup, the one-RTT delay in pipeline due to probing does not create idle periods: a previously-winning sender takes roughly one RTT to observe it has lost (via ACK feedback) and stop sending, while a newly-winning sender similarly takes one RTT to observe it has won and resume sending. \looseness=-1

\para{Sender rate control.}
Switches do not drop packets or explicitly throttle losing flows. They only update auction flags in packet headers. Rate control is enforced at the sender. A sender transmits data at line rate only when the \textit{Auction status} observed in the most recent probe ACK is set : indicating that the flow won all link auctions along its path.
Regardless of win/loss, probes are sent continuously (ACK-clocked) once per RTT to register latest bids and obtain fresh outcome status. \looseness=-1

\para{Queuing.}
Switches enqueue data packets in a FIFO queue. Since senders gate transmission on end-to-end wins, steady-state queue occupancy is bounded by the link BDP. Short bursts may occur during temporary overcommitment changes, but queues drain once blocked flows are excluded or contention clears. If the queue explodes, switches clear the winners' list forbiding everyone to send for an RTT to clear queues. Probes and ACKs are placed in a separate FIFO to avoid being delayed behind data packets.

\para{Switch Complexity.}
Although \name requires extensions to the switch pipeline, the added state and logic are lightweight compared to datapaths like pFabric. Each egress port maintains top-$k$ winner records for the current and next RTT epochs, along with simple comparators to update bids and auction-status flags. Let $E$ denote the number of egress ports and $k$ the winner quota per RTT. The required switch state is $O(E\cdot k)$ words for per-port winner lists, plus $O(E)$ registers for configuration and timers. In practice, $k$ is small especially in oversubscribed topologies. Figure~\ref{fig:switch_auction} illustrates the resulting switch block diagram.

\begin{figure}
    \centering
    \includegraphics[width=0.42\textwidth]{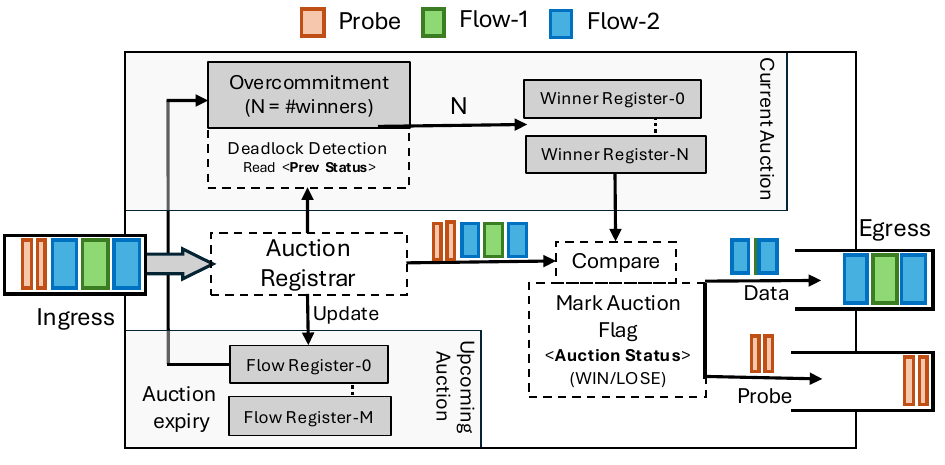}
    \caption{Block diagram of switch for holding auctions. }
    \label{fig:switch_auction}
\end{figure}
\begin{figure}
    \centering
    \includegraphics[width=0.35\textwidth]{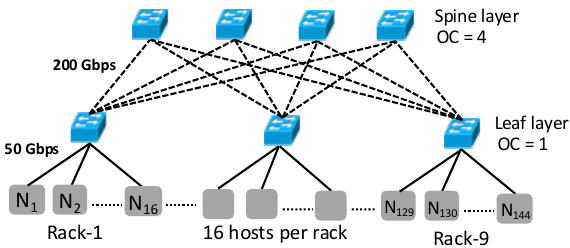}
    \caption{Experiment Topology(OC:Overcommitment)}
    \label{fig:ls_topology}
\end{figure}
\begin{table}
\scriptsize
\centering
\linespread{1.05}\selectfont\centering
\begin{tabular}{|p{1.5cm}|p{2.8cm}|p{2cm}|}
\hline
\textbf{Distribution}       & \textbf{Apps/Workloads}    & \textbf{Objective}      \\
\hline
Websearch       & Web, RPC    & Latency/FCT      \\ \hline
Datamining      & Analytics, Control           & Latency/FCT      \\ \hline
Uniform            & Map-Reduce, Background  & Deadline, Fairness         \\ \hline
Alibaba         & Storage    & Best Effort       \\ \hline
\end{tabular}
\caption{Evaluation Workloads and Flow Size Distributions.}
\vspace{-0.1cm}
\label{tab:workloads}
\end{table}
\para{Routing.}
The current \name design is incompatible with per-packet load balancing (e.g., packet spraying), but works naturally with path-stable schemes such as flow-based ECMP. While \name exposes link prices that could serve as congestion signals for path selection, we do not explore congestion-aware routing in this paper. Investigating how probes can be used to explore alternative paths while sending data on currently chosen paths is scoped for future work.

\section{Evaluation}

\label{sec:eval}
In this section, we evaluate \name through packet-level simulations in ns-3. While \name's market abstraction applies broadly to any network stack, our prototype is implemented on top of the TCP stack, enabling direct comparison against widely deployed datacenter transport and scheduling baselines.
Our evaluation answers three key questions: \vspace{-0.1cm}
\begin{enumerate}
    \item Does \name reproduce state-of-the-art \emph{single objective} schedulers when all flows share the same goal?
    \item Does \name enable \emph{efficient coexistence} of heterogeneous objectives under realistic load?
    \item Does the market remain stable and adaptive under workload drift, bursts, and varying degrees of contention?
\end{enumerate}

\para{Topology.}
We simulate a 144-host leaf-spine topology widely used in prior scheduling work~\cite{pfabric, pias, karuna}. The network consists of 9 racks with 16 hosts per rack. Host links operate at 50\,Gbps, and spine links at 200\,Gbps, yielding full bisection bandwidth (1:1 oversubscription). Traffic is balanced using flow-level hash ECMP. Cross-rack RTT is approximately 10\,$\mu$s.
Figure~\ref{fig:ls_topology} illustrates the topology.

\para{Workloads.}
We model representative datacenter traffic using established workload families used in prior works~\cite{dctcp, pfabric, homa, pias, d2tcp, hpcc} summarized in Table~\ref{tab:workloads}. The uniform workload draws flow sizes uniformly at random from the range $[50\,\mathrm{kB},\,10\,\mathrm{MB}]$. These workloads span latency-sensitive RPCs, web, deadline-driven transfers, coflow-style shuffle/map-reduce traffic, and best-effort storage flows. Unless stated otherwise, a flow's source and destination hosts are sampled uniformly at random except for explicit burst/hotspot/coflow scenarios. Flow arrivals are Poisson except for flows belonging to a coflow where they are concurrent. Inter-arrival times are on a microsecond scale to match the link capacity of 50 Gbps at high load. 
In our dynamic experiments, we synthetically create random drift to create variation in load (diurnal), workload mix, and objective mix to model the datacenter environment and study \name's ability to adapt.

\para{Baselines.}
In all experiments, we compare \name against strong online baselines for each objective: (i) TCP Reno and DCTCP for fair sharing, (ii) D2TCP for deadline awareness, (iii) pFabric as an idealized SRPT scheduler, (iv) 8-queue Priority Queueing as an SRPT-style restricted-datapath baseline, and (v) Sincronia for coflow completion time.

For Priority Queueing, we recompute queue thresholds before each run using flow-size percentiles to ensure accurate priority mappings. Since Sincronia lacks ns-3 support, we replay the same ns-3 arrival trace in its flow-level simulator to obtain relative flow ordering. We then use this flow ordering as pFabric priorities in ns-3 in the experiment. Sincronia solves a centralized ILP once per millisecond, scheduling all flows that arrived since the previous update. To favor Sincronia, we assume instantaneous communication and that the ILP completes within 1 ms, although such centralized optimization can incur substantially higher overhead at datacenter scale.

\begin{figure}[t]
    \centering
    \includegraphics[width=0.45\textwidth]{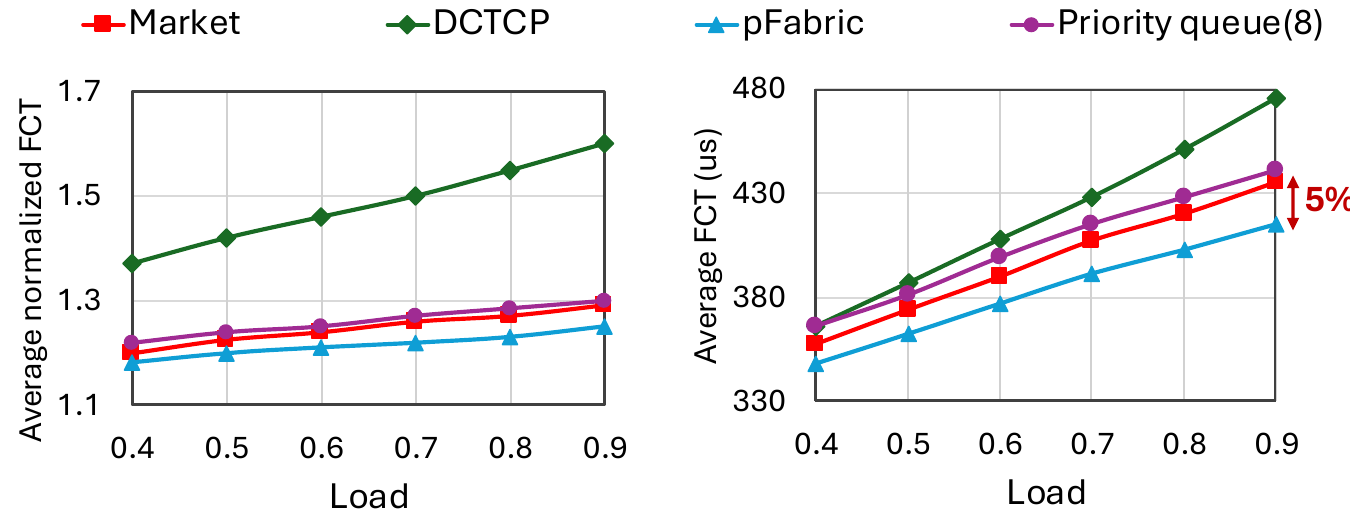}
    \vspace{-0.1cm}
    \caption{Single Policy Benchmark : Average FCT.}
    \label{fig:srpt_single}
\end{figure}
\vspace{0.1cm}
\begin{figure}[t]
    \centering
    \includegraphics[width=0.45\textwidth]{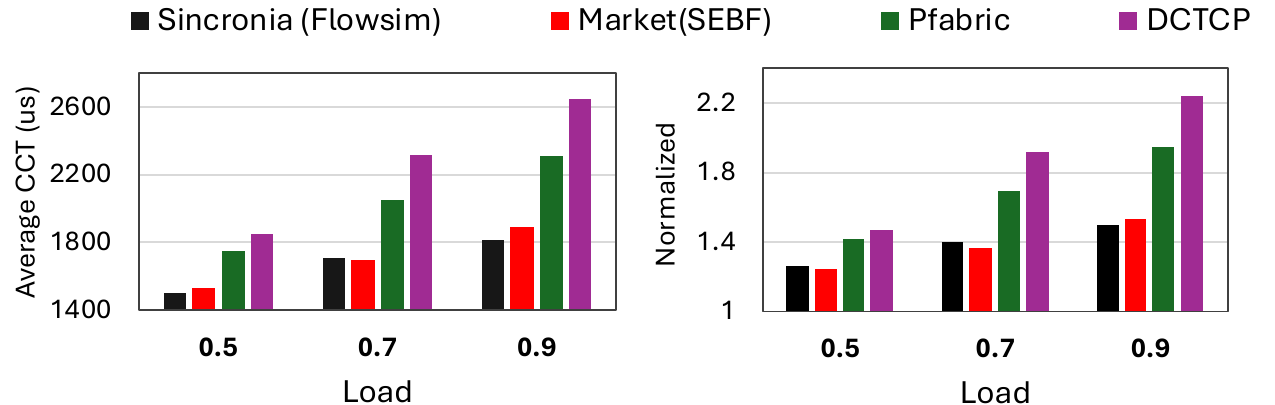}
    \vspace{-0.1cm}
    \caption{Single Policy Benchmark : Average CCT}
    \vspace{0.2cm}
    \label{fig:single_cflows}
\end{figure}
\begin{figure}[t]
    \centering
    \includegraphics[width=0.45\textwidth]{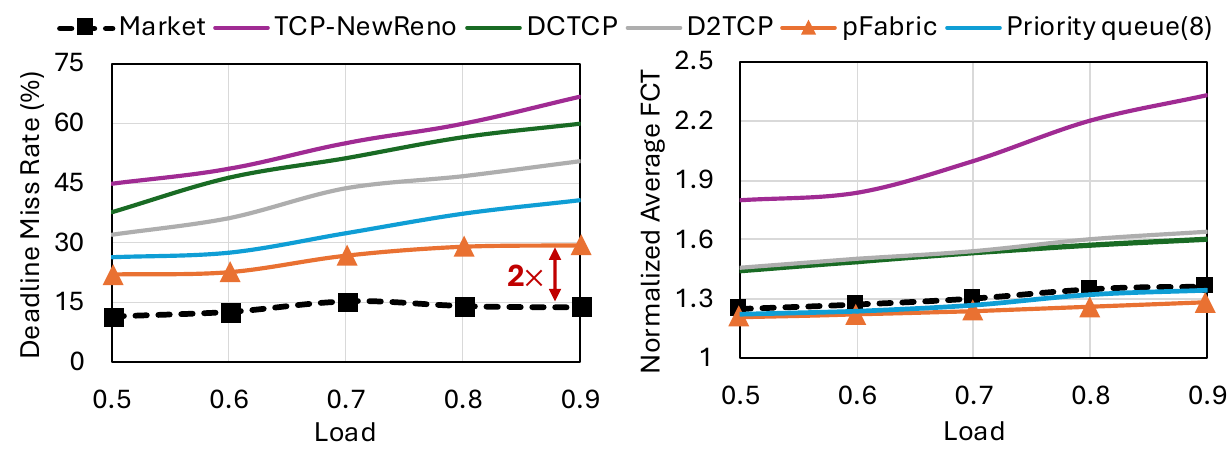}
    \caption{\textbf{Mixed-objective (deadline miss rate, average FCT)} in heterogeneous traffic mix. \name's deadline miss rate less than 15\% at all loads while average FCT is same as pFabric's }
    \label{fig:srpt_edf_2}
\end{figure}

\begin{figure}[t]
    \centering
    \includegraphics[width=0.45\textwidth]{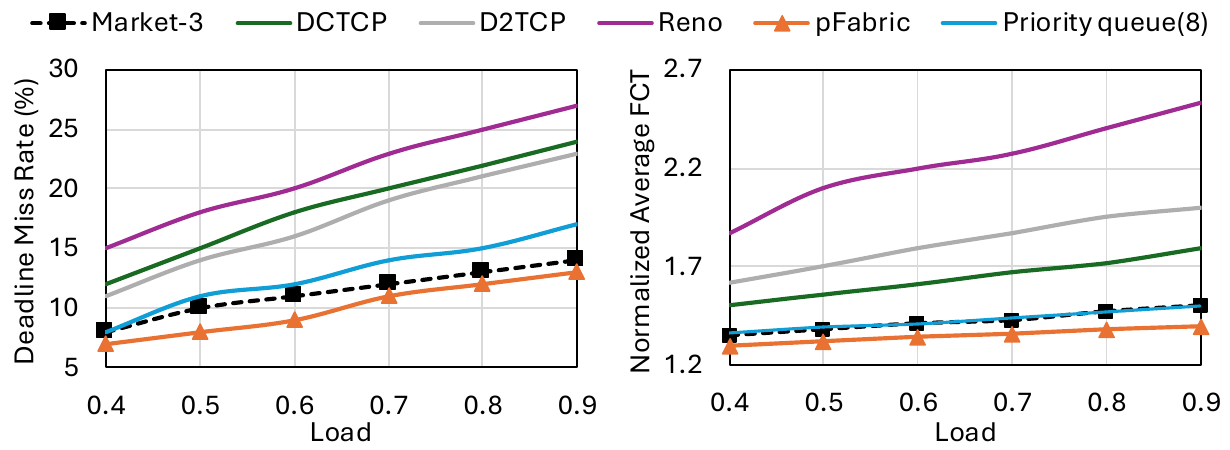}
    \caption{\textbf{Mixed-objective (deadline miss rate, average FCT)} in homogeneous traffic mix for varying load. \name is within 5\% of pFabric on both metrics even with a FIFO datapath.}
    \label{fig:srpt_edf_homo}
\end{figure}
\vspace{-0.2cm}
\subsection{Single-objective benchmarks}
\begin{figure*}
    \centering
    \includegraphics[width=0.95\textwidth]{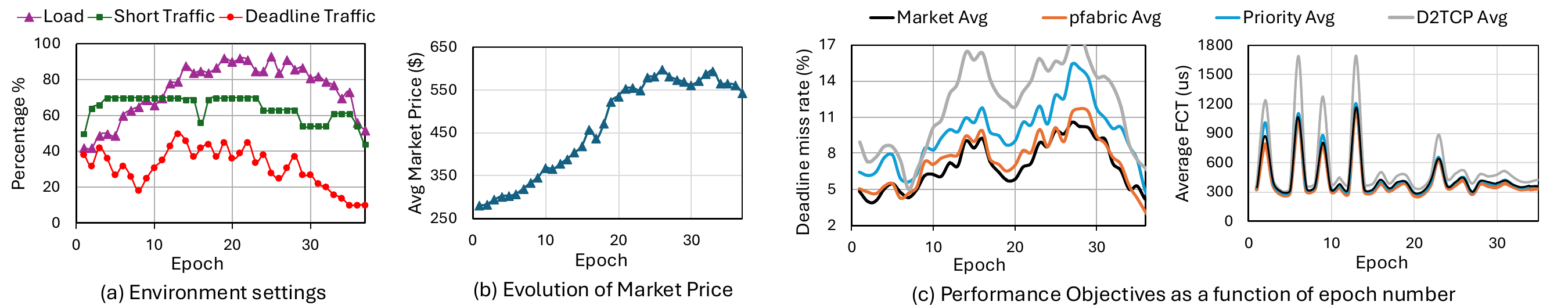}
    \caption{\textbf{Workload drift and mixed-objective performance (deadline miss rate, and average FCT) over 40 epochs.} \name achieves lowest deadline miss rate and average FCT of latency-sensitive traffic over all epochs. Market is stable and prices follow the drift in load.}
    \vspace{-0.2cm}
    \label{fig:srpt_edf_dynamic}
\end{figure*}
We first validate that \name can realize canonical single-objective schedulers when all flows share the same goal and same bidding API in \name.

\para{Latency-sensitive flows (SRPT).}
We generate websearch short traffic with Poisson arrivals and sweep load while measuring the average FCT of around 5000 flows. Flows arrive in bursts and use \name’s latency-driven bidding API (Appendix~\S\ref{sec:app_api}). We compare against pFabric, 8-queue Priority Queueing, and DCTCP. As shown in Figure~\ref{fig:srpt_single}, \name outperforms Priority Queueing and DCTCP, and remains within 5\% of pFabric despite its single FIFO datapath compared to pFabric’s idealized priority scheduling. This small gap is largely due to datapath differences and \name’s probe-based auction clearance. \looseness=-1

\para{Coflows.}
We use the Shortest Effective Bottleneck First (SEBF) bidding API for coflows. Each coflow in our synthetic workload contains a random number of concurrent flows with randomly sampled source--destination host pairs.
Flow sizes are drawn from the Uniform distribution in Table~\ref{tab:workloads}.

We compare \name against pFabric (SRPT), DCTCP (fair sharing), and Sincronia. Figure~\ref{fig:single_cflows} reports average coflow completion time (CCT) on (a) an absolute scale and (b) a normalized scale. 
Sincronia flow-level simulator solves a centralized ILP on  millisecond scale, while \name uses the lightweight app-centric bidding at RTT granularity (Appendix~\S\ref{sec:app_api}). Still \name performs similarly, sometimes better compared to Sincronia at all loads while significantly outperforming pFabric and DCTCP.

\subsection{Multi-objective benchmarks}
\begin{figure*}
    \centering
    \includegraphics[width=0.85\textwidth]
    {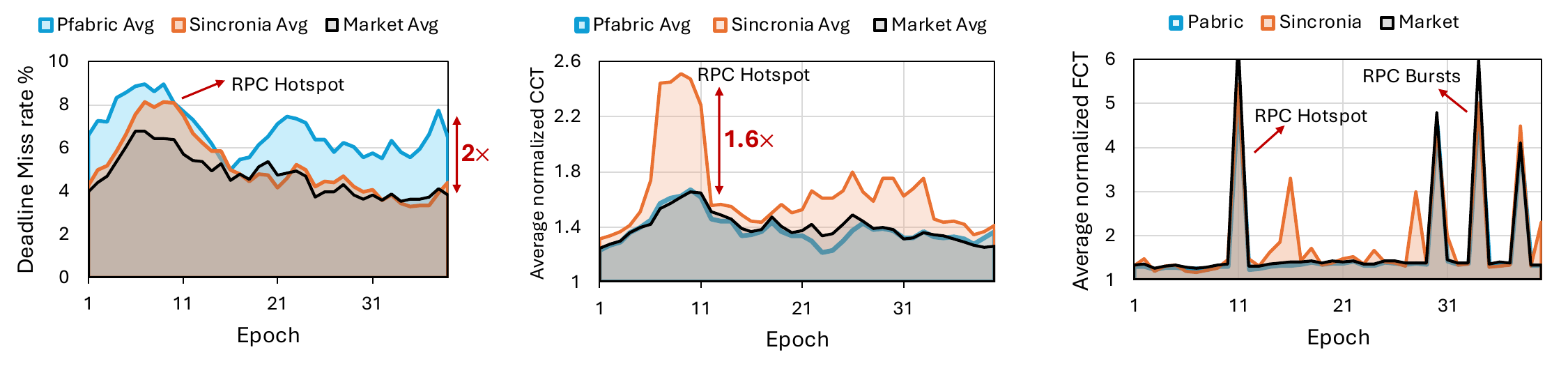}
    \vspace{-0.3cm}
    \caption{\textbf{Mixed-objective performance(deadline miss rate, average FCT and average CCT) over 40 epochs of drift.} \name achieves \textbf{2}$\times$ lower deadline miss rate and \textbf{1.6}$\times$ lower average CCT compared to pFabric and Sincronia with low average FCT for short flows even in bursts.}
    \vspace{-0.2cm}
    \label{fig:mixedcflow_dynamic}
\end{figure*}

\begin{figure}
    \centering
    \includegraphics[width=0.45\textwidth]{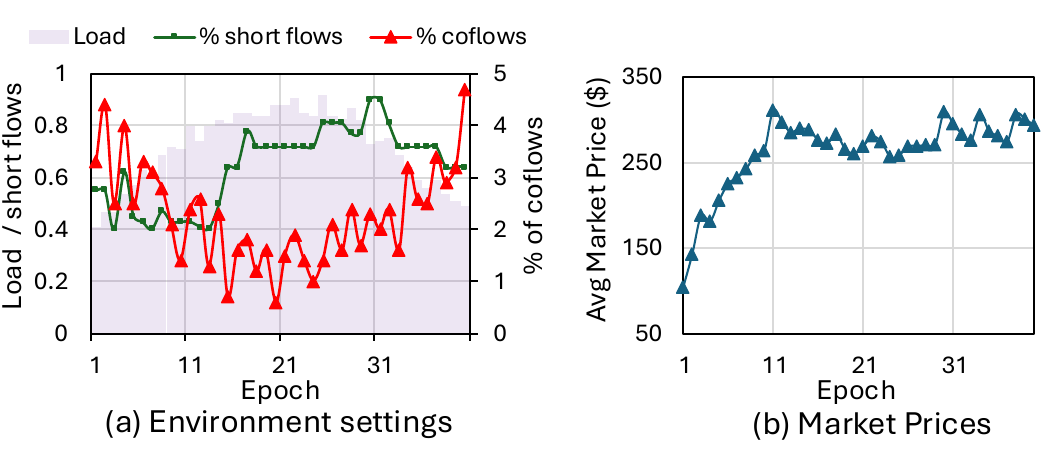}
    \caption{\textbf{Workload drift and market prices over 40 epochs.}}
    \label{fig:cflowmix_set}
\end{figure}

We next evaluate whether \name simultaneously delivers good performance across incompatible objectives.

\subsubsection{Latency-sensitive + Deadline-sensitive traffic.}
We consider mixes of latency-sensitive flows (objective: minimize average FCT) and deadline-sensitive flows (objective: minimize deadline miss rate). We evaluate three environments: (i) a heterogeneous mix designed to stress objective mismatch, (ii) a homogeneous random mix with many ultra-short flows to test datapath, and (iii) a drifting workload that changes load and mix over time.

\para{Heterogeneous mix.}
We generate a heterogeneous workload with 85\% short latency-sensitive flows sampled from the websearch distribution and 15\% deadline flows sampled from the Uniform distribution with randomly assigned slack (from 100 us to 1ms). We run simulations of \textasciitilde8000 flows at each load. We report metrics for the last \textasciitilde5000 flows, after \name's bidding agents have warmed up and the market prices have converged for the given load.

Figure~\ref{fig:srpt_edf_2} reports deadline miss rate and average FCT as load increases. \name matches pFabric’s normalized FCT while consistently achieving a much lower deadline miss rate than pFabric, Priority Queueing, and D2TCP. At high load, \name reduces misses by roughly $2\times$ versus pFabric, $3\times$ versus Priority Queueing, and $3.5\times$ versus D2TCP. This reflects objective mismatch in SRPT-style baselines and rate-control limits in D2TCP. \name's bidding API allows deadline flows to bid high as slack shrinks while latency flows dominate otherwise.

\para{Homogeneous random mix.}
We next consider a more homogeneous workload environment in which flow sizes are sampled across three distributions : websearch, datamining, and Uniform, with each contributing roughly 33\% of flows. This produces a workload mix dominated by ultra-short flows (\textasciitilde85\% shorter than 1 RTT), which strongly benefits schedulers with a powerful datapath (e.g. pFabric). We assign objectives independently at random: roughly half the flows are deadline-sensitive (with randomly sampled deadline slack) and half are latency-sensitive. Each run contains \textasciitilde8000 flows as before. Figure~\ref{fig:srpt_edf_homo} reports deadline miss rate and average FCT of the last 5000 flows. Despite being constrained to a FIFO datapath and incurring an additional RTT for auction probing on many flows, \name remains within \textasciitilde5\% of pFabric in average normalized FCT and within \textasciitilde1\% in deadline miss rate at high load, while substantially outperforming priority queueing, DCTCP, TCP-Reno, and D2TCP.

\para{Drift and varying workloads.}
Finally, we evaluate \name under dynamic datacenter drift using Websearch, Datamining, and Uniform workloads. Over 40 epochs, we vary offered load diurnally with noise, and drift both the fraction of short flows and the mix of deadline vs. latency objectives via a sticky random walk. To model bursts and hotspots, \textasciitilde15\% of epochs inject a burst of short flows (at $2\times$ the load), sometimes concentrated in a few racks to create incast hotspots. Each epoch is a simulation run containing thousands of flows. \name updates bidding policies once at the start of every epoch using the previous epoch’s global price ditribution, and we report performance on all flows without a separate warm-up of \name. In fact, we show that market prices track the slow drift perfectly due to timescale separation.

Figure~\ref{fig:srpt_edf_dynamic}(a,b) shows the resulting drift process with load, short traffic, and deadline traffic varying over epochs and market prices following the drift. Figure~\ref{fig:srpt_edf_dynamic}(b) shows that \name consistently achieves the lowest deadline miss rate across epochs, even beating pFabric, while maintaining average normalized FCT almost as pFabric, despite the majority traffic being short flows and even during sudden bursts and hotspots. \name also outperforms other baselines including priority queues, D2TCP on both objectives in all epochs.

\subsubsection{Latency + Deadline + Coflow traffic.}

We now add coflow traffic to the earlier mix of deadline and latency-sensitive traffic, and re-run the drift experiment over diurnally varying load, bursts, hotspots and varying percentages of short flows, deadline-sensitive flows, and coflows over a period of 40 epochs.Again, the bidding agents update their policies once at start of every epoch. Coflow sizes are sampled from Uniform distribution. 
The flows of a coflow are concurrent, but the source-destination pairs are sampled randomly. Note that a single coflow gives rise to tens and hundreds of flows, so the percentage of coflows in every epoch is less than 10\% of the total workload to ensure a good mix with other traffic. 
The datacenter drift specs and accordingly varying market prices are shown in Figure~\ref{fig:cflowmix_set}(a) and (b) respectively.

Figure~\ref{fig:mixedcflow_dynamic} shows that \name achieves lowest deadline miss rate (2$\times$ less than pFabric) outperforming both Sincronia and pFabric across epochs. In addition, \name also obtains low average flow completion time (FCT) for short flows and coflow completion time (CCT) for coflows simultaneously. Surprisingly, Sincronia---a coflow scheduling regimen---fails to reduce the average CCT for coflows. The reason is simple: Sincronia is designed for a coflow-only regime; the coflow abstraction considers a single flow as a coflow of size 1 and gives similar weight to all such "coflows". Although considering the completion time of all the flows (regardless of objective) Sincronia is near-optimal, but it delivers high average CCT compared to \name and pFabric for the actual multi-flow coflows. Sincronia also has occasional spikes in average FCT of short flows, especially during bursts party due to ms granularity of its flow-level simulator. \name performs best for all three metrics.
\vspace{-0.3cm}
\subsubsection{Latency-sensitive + Best-effort + Fairness traffic.}
We now mix latency-sensitive flows (objective: minimize average FCT) sampled from Websearch, fairness flows (objective: min bandwidth 5\%) sampled from Uniform and low-tier best effort flows (no objective) sampled from Alibaba, and evaluate the objectives performance at a high load of 90\%. To demonstrate the objective mismatch in baselines, we sweep the percentage of aggregate fairness + best effort traffic from low to high. Figure~\ref{fig:srpt_fair} plots the tail bandwidth achieved by fairness-sensitive flows and the average FCT of short flows. \name ensures that fairness flows meet their minimum bandwidth SLO while preserving low latency for short flows. In contrast, SRPT-based schedulers exhibit increasing average FCT across the sweep due to yielding to fairness and best-effort flows and objective mismatch at high \% of fairness and best effort traffic.\vspace{-0.3cm}
\begin{figure}[h]
    \centering
    \includegraphics[width=0.47\textwidth]{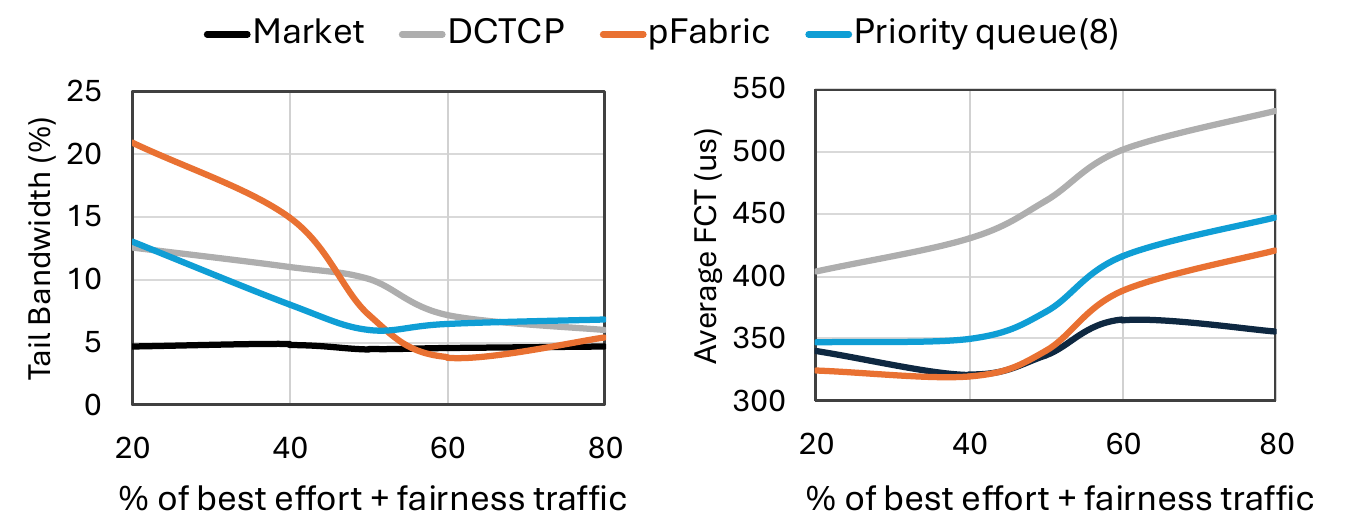}
    \caption{\textbf{Tail bandwidth and average FCT} in mix of latency-sensitive, fairness (5\% guaranteed throughput) and best-effort traffic.}
    \vspace{-0.2cm}
    \label{fig:srpt_fair}
\end{figure}

\vspace{0.3cm}
\section{Conclusion}
\label{sec:conclusion}
This paper introduced \name, a market-based scheduling primitive that enables the co-existence of heterogeneous objectives through lightweight, repeated per-link auctions. Flows encode objective-specific \emph{urgency} in their bids, while budgets capture their overall \emph{entitlement}. \name not only outperforms state-of-the-art schedulers, but also provides a flexible foundation for future objectives.

\para{Acknowledgements.}
This research was supported by NSF FMitF-2421734, NSF CAREER-2144766, NSF PPoSS-2217099, NSF CNS-2211382, NSF CNS-2403026,  Jane Street, and Sloan fellowship FG-2022-18504.
\begin{center}
\emph{In memory of Prateesh Goyal}
\end{center}

\label{bodypage}

\bibliographystyle{ACM-Reference-Format}
\begin{small}
\bibliography{refs}


\begin{thebibliography}{55}


\ifx \showCODEN    \undefined \def \showCODEN     #1{\unskip}     \fi
\ifx \showISBNx    \undefined \def \showISBNx     #1{\unskip}     \fi
\ifx \showISBNxiii \undefined \def \showISBNxiii  #1{\unskip}     \fi
\ifx \showISSN     \undefined \def \showISSN      #1{\unskip}     \fi
\ifx \showLCCN     \undefined \def \showLCCN      #1{\unskip}     \fi
\ifx \shownote     \undefined \def \shownote      #1{#1}          \fi
\ifx \showarticletitle \undefined \def \showarticletitle #1{#1}   \fi
\ifx \showURL      \undefined \def \showURL       {\relax}        \fi
\providecommand\bibfield[2]{#2}
\providecommand\bibinfo[2]{#2}
\providecommand\natexlab[1]{#1}
\providecommand\showeprint[2][]{arXiv:#2}

\bibitem[Agarwal et~al\mbox{.}(2018)]%
        {sincronia}
\bibfield{author}{\bibinfo{person}{Saksham Agarwal}, \bibinfo{person}{Shijin Rajakrishnan}, \bibinfo{person}{Akshay Narayan}, \bibinfo{person}{Rachit Agarwal}, \bibinfo{person}{David Shmoys}, {and} \bibinfo{person}{Amin Vahdat}.} \bibinfo{year}{2018}\natexlab{}.
\newblock \showarticletitle{Sincronia: near-optimal network design for coflows}. In \bibinfo{booktitle}{\emph{Proceedings of the 2018 Conference of the ACM Special Interest Group on Data Communication}} (Budapest, Hungary) \emph{(\bibinfo{series}{SIGCOMM '18})}. \bibinfo{publisher}{Association for Computing Machinery}, \bibinfo{address}{New York, NY, USA}, \bibinfo{pages}{16–29}.
\newblock
\showISBNx{9781450355674}
\href{https://doi.org/10.1145/3230543.3230569}{doi:\nolinkurl{10.1145/3230543.3230569}}


\bibitem[Al-Fares et~al\mbox{.}(2010)]%
        {hedera}
\bibfield{author}{\bibinfo{person}{Mohammad Al-Fares}, \bibinfo{person}{Sivasankar Radhakrishnan}, \bibinfo{person}{Barath Raghavan}, \bibinfo{person}{Nelson Huang}, {and} \bibinfo{person}{Amin Vahdat}.} \bibinfo{year}{2010}\natexlab{}.
\newblock \showarticletitle{Hedera: dynamic flow scheduling for data center networks}. In \bibinfo{booktitle}{\emph{Proceedings of the 7th USENIX Conference on Networked Systems Design and Implementation}} (San Jose, California) \emph{(\bibinfo{series}{NSDI'10})}. \bibinfo{publisher}{USENIX Association}, \bibinfo{address}{USA}, \bibinfo{pages}{19}.
\newblock


\bibitem[Alizadeh et~al\mbox{.}(2014)]%
        {conga}
\bibfield{author}{\bibinfo{person}{Mohammad Alizadeh}, \bibinfo{person}{Tom Edsall}, \bibinfo{person}{Suman Dharmapurikar}, \bibinfo{person}{Rahul Vaidya}, \bibinfo{person}{Kevin Chu}, \bibinfo{person}{Andy Fingerhut}, \bibinfo{person}{Vinh~The Lam}, \bibinfo{person}{Francis Matus}, \bibinfo{person}{Rong Pan}, \bibinfo{person}{Nandita Yadav}, {and} \bibinfo{person}{George Varghese}.} \bibinfo{year}{2014}\natexlab{}.
\newblock \showarticletitle{CONGA: Distributed Congestion-Aware Load Balancing for Datacenters}. In \bibinfo{booktitle}{\emph{Proceedings of the 2014 ACM Conference on SIGCOMM}} \emph{(\bibinfo{series}{SIGCOMM '14})}. \bibinfo{publisher}{Association for Computing Machinery}, \bibinfo{address}{New York, NY, USA}.
\newblock


\bibitem[Alizadeh et~al\mbox{.}(2010)]%
        {dctcp}
\bibfield{author}{\bibinfo{person}{Mohammad Alizadeh}, \bibinfo{person}{Albert Greenberg}, \bibinfo{person}{David~A. Maltz}, \bibinfo{person}{Jitendra Padhye}, \bibinfo{person}{Parveen Patel}, \bibinfo{person}{Balaji Prabhakar}, \bibinfo{person}{Sudipta Sengupta}, {and} \bibinfo{person}{Murari Sridharan}.} \bibinfo{year}{2010}\natexlab{}.
\newblock \showarticletitle{Data center TCP (DCTCP)}. In \bibinfo{booktitle}{\emph{Proceedings of the ACM SIGCOMM 2010 Conference}} (New Delhi, India) \emph{(\bibinfo{series}{SIGCOMM '10})}. \bibinfo{publisher}{Association for Computing Machinery}, \bibinfo{address}{New York, NY, USA}, \bibinfo{pages}{63–74}.
\newblock
\showISBNx{9781450302012}
\href{https://doi.org/10.1145/1851182.1851192}{doi:\nolinkurl{10.1145/1851182.1851192}}


\bibitem[Alizadeh et~al\mbox{.}(2013)]%
        {pfabric}
\bibfield{author}{\bibinfo{person}{Mohammad Alizadeh}, \bibinfo{person}{Shuang Yang}, \bibinfo{person}{Milad Sharif}, \bibinfo{person}{Sachin Katti}, \bibinfo{person}{Nick McKeown}, \bibinfo{person}{Balaji Prabhakar}, {and} \bibinfo{person}{Scott Shenker}.} \bibinfo{year}{2013}\natexlab{}.
\newblock \showarticletitle{pFabric: minimal near-optimal datacenter transport}. In \bibinfo{booktitle}{\emph{Proceedings of the ACM SIGCOMM 2013 Conference on SIGCOMM}} (Hong Kong, China) \emph{(\bibinfo{series}{SIGCOMM '13})}. \bibinfo{publisher}{Association for Computing Machinery}, \bibinfo{address}{New York, NY, USA}, \bibinfo{pages}{435–446}.
\newblock
\showISBNx{9781450320566}
\href{https://doi.org/10.1145/2486001.2486031}{doi:\nolinkurl{10.1145/2486001.2486031}}


\bibitem[Bai et~al\mbox{.}(2014)]%
        {pias}
\bibfield{author}{\bibinfo{person}{Wei Bai}, \bibinfo{person}{Li Chen}, \bibinfo{person}{Kai Chen}, \bibinfo{person}{Dongsu Han}, \bibinfo{person}{Chen Tian}, {and} \bibinfo{person}{Weicheng Sun}.} \bibinfo{year}{2014}\natexlab{}.
\newblock \showarticletitle{PIAS: Practical Information-Agnostic Flow Scheduling for Data Center Networks}. In \bibinfo{booktitle}{\emph{Proceedings of the 13th ACM Workshop on Hot Topics in Networks}} (Los Angeles, CA, USA) \emph{(\bibinfo{series}{HotNets-XIII})}. \bibinfo{publisher}{Association for Computing Machinery}, \bibinfo{address}{New York, NY, USA}, \bibinfo{pages}{1–7}.
\newblock
\showISBNx{9781450332569}
\href{https://doi.org/10.1145/2670518.2673871}{doi:\nolinkurl{10.1145/2670518.2673871}}


\bibitem[Bitsaki et~al\mbox{.}(2005)]%
        {psp_3}
\bibfield{author}{\bibinfo{person}{Marina Bitsaki}, \bibinfo{person}{George~D. Stamoulis}, {and} \bibinfo{person}{Costas Courcoubetis}.} \bibinfo{year}{2005}\natexlab{}.
\newblock \showarticletitle{A new strategy for bidding in the network-wide progressive second price auction for bandwidth}. In \bibinfo{booktitle}{\emph{Proceedings of the 2005 ACM Conference on Emerging Network Experiment and Technology}} (Toulouse, France) \emph{(\bibinfo{series}{CoNEXT '05})}. \bibinfo{publisher}{Association for Computing Machinery}, \bibinfo{address}{New York, NY, USA}, \bibinfo{pages}{146–155}.
\newblock
\showISBNx{159593197X}
\href{https://doi.org/10.1145/1095921.1095941}{doi:\nolinkurl{10.1145/1095921.1095941}}


\bibitem[Buckley et~al\mbox{.}(2025)]%
        {qos_dngrade}
\bibfield{author}{\bibinfo{person}{Matthew Buckley}, \bibinfo{person}{Parsa Pazhooheshy}, \bibinfo{person}{Z.~Morley Mao}, \bibinfo{person}{Nandita Dukkipati}, \bibinfo{person}{Hamid~Hajabdolali Bazzaz}, \bibinfo{person}{Priyaranjan Jha}, \bibinfo{person}{Yingjie Bi}, \bibinfo{person}{Steve Middlekauff}, {and} \bibinfo{person}{Yashar Ganjali}.} \bibinfo{year}{2025}\natexlab{}.
\newblock \showarticletitle{Learnings from Deploying Network {QoS} Alignment to Application Priorities for Storage Services}. In \bibinfo{booktitle}{\emph{22nd USENIX Symposium on Networked Systems Design and Implementation (NSDI 25)}}. \bibinfo{publisher}{USENIX Association}, \bibinfo{address}{Philadelphia, PA}, \bibinfo{pages}{37--53}.
\newblock
\showISBNx{978-1-939133-46-5}
\urldef\tempurl%
\url{https://www.usenix.org/conference/nsdi25/presentation/buckley}
\showURL{%
\tempurl}


\bibitem[Cao et~al\mbox{.}(2024)]%
        {crux}
\bibfield{author}{\bibinfo{person}{Jiamin Cao}, \bibinfo{person}{Yu Guan}, \bibinfo{person}{Kun Qian}, \bibinfo{person}{Jiaqi Gao}, \bibinfo{person}{Wencong Xiao}, \bibinfo{person}{Jianbo Dong}, \bibinfo{person}{Binzhang Fu}, \bibinfo{person}{Dennis Cai}, {and} \bibinfo{person}{Ennan Zhai}.} \bibinfo{year}{2024}\natexlab{}.
\newblock \showarticletitle{Crux: GPU-Efficient Communication Scheduling for Deep Learning Training}. In \bibinfo{booktitle}{\emph{Proceedings of the ACM SIGCOMM 2024 Conference}} (Sydney, NSW, Australia) \emph{(\bibinfo{series}{ACM SIGCOMM '24})}. \bibinfo{publisher}{Association for Computing Machinery}, \bibinfo{address}{New York, NY, USA}, \bibinfo{pages}{1–15}.
\newblock
\showISBNx{9798400706141}
\href{https://doi.org/10.1145/3651890.3672239}{doi:\nolinkurl{10.1145/3651890.3672239}}


\bibitem[Chen et~al\mbox{.}(2016)]%
        {karuna}
\bibfield{author}{\bibinfo{person}{Li Chen}, \bibinfo{person}{Kai Chen}, \bibinfo{person}{Wei Bai}, {and} \bibinfo{person}{Mohammad Alizadeh}.} \bibinfo{year}{2016}\natexlab{}.
\newblock \showarticletitle{Scheduling Mix-flows in Commodity Datacenters with Karuna}. In \bibinfo{booktitle}{\emph{Proceedings of the 2016 ACM SIGCOMM Conference}} (Florianopolis, Brazil) \emph{(\bibinfo{series}{SIGCOMM '16})}. \bibinfo{publisher}{Association for Computing Machinery}, \bibinfo{address}{New York, NY, USA}, \bibinfo{pages}{174–187}.
\newblock
\showISBNx{9781450341936}
\href{https://doi.org/10.1145/2934872.2934888}{doi:\nolinkurl{10.1145/2934872.2934888}}


\bibitem[Chowdhury and Stoica(2015)]%
        {aalo}
\bibfield{author}{\bibinfo{person}{Mosharaf Chowdhury} {and} \bibinfo{person}{Ion Stoica}.} \bibinfo{year}{2015}\natexlab{}.
\newblock \showarticletitle{Efficient Coflow Scheduling Without Prior Knowledge}. In \bibinfo{booktitle}{\emph{Proceedings of the 2015 ACM Conference on Special Interest Group on Data Communication}} (London, United Kingdom) \emph{(\bibinfo{series}{SIGCOMM '15})}. \bibinfo{publisher}{Association for Computing Machinery}, \bibinfo{address}{New York, NY, USA}, \bibinfo{pages}{393–406}.
\newblock
\showISBNx{9781450335423}
\href{https://doi.org/10.1145/2785956.2787480}{doi:\nolinkurl{10.1145/2785956.2787480}}


\bibitem[Chowdhury et~al\mbox{.}(2014)]%
        {varys}
\bibfield{author}{\bibinfo{person}{Mosharaf Chowdhury}, \bibinfo{person}{Yuan Zhong}, {and} \bibinfo{person}{Ion Stoica}.} \bibinfo{year}{2014}\natexlab{}.
\newblock \showarticletitle{Efficient coflow scheduling with Varys}.
\newblock \bibinfo{journal}{\emph{SIGCOMM Comput. Commun. Rev.}} \bibinfo{volume}{44}, \bibinfo{number}{4} (\bibinfo{date}{aug} \bibinfo{year}{2014}), \bibinfo{pages}{443–454}.
\newblock
\showISSN{0146-4833}
\href{https://doi.org/10.1145/2740070.2626315}{doi:\nolinkurl{10.1145/2740070.2626315}}


\bibitem[Cocchi et~al\mbox{.}(1991)]%
        {price_priority}
\bibfield{author}{\bibinfo{person}{Ron Cocchi}, \bibinfo{person}{Deborah Estrin}, \bibinfo{person}{Scott Shenker}, {and} \bibinfo{person}{Lixia Zhang}.} \bibinfo{year}{1991}\natexlab{}.
\newblock \showarticletitle{A study of priority pricing in multiple service class networks}. In \bibinfo{booktitle}{\emph{Conference on Applications, Technologies, Architectures, and Protocols for Computer Communication}}.
\newblock
\urldef\tempurl%
\url{https://api.semanticscholar.org/CorpusID:13446820}
\showURL{%
\tempurl}


\bibitem[Demers et~al\mbox{.}(1989a)]%
        {fairness_analysis}
\bibfield{author}{\bibinfo{person}{A. Demers}, \bibinfo{person}{S. Keshav}, {and} \bibinfo{person}{S. Shenker}.} \bibinfo{year}{1989}\natexlab{a}.
\newblock \showarticletitle{Analysis and simulation of a fair queueing algorithm}. In \bibinfo{booktitle}{\emph{Symposium Proceedings on Communications Architectures \&amp; Protocols}} (Austin, Texas, USA) \emph{(\bibinfo{series}{SIGCOMM '89})}. \bibinfo{publisher}{Association for Computing Machinery}, \bibinfo{address}{New York, NY, USA}, \bibinfo{pages}{1–12}.
\newblock
\showISBNx{0897913329}
\href{https://doi.org/10.1145/75246.75248}{doi:\nolinkurl{10.1145/75246.75248}}


\bibitem[Demers et~al\mbox{.}(1989b)]%
        {fair_queueing}
\bibfield{author}{\bibinfo{person}{A. Demers}, \bibinfo{person}{S. Keshav}, {and} \bibinfo{person}{S. Shenker}.} \bibinfo{year}{1989}\natexlab{b}.
\newblock \showarticletitle{Analysis and simulation of a fair queueing algorithm}.
\newblock \bibinfo{journal}{\emph{SIGCOMM Comput. Commun. Rev.}} \bibinfo{volume}{19}, \bibinfo{number}{4} (\bibinfo{date}{Aug.} \bibinfo{year}{1989}), \bibinfo{pages}{1–12}.
\newblock
\showISSN{0146-4833}
\href{https://doi.org/10.1145/75247.75248}{doi:\nolinkurl{10.1145/75247.75248}}


\bibitem[Di et~al\mbox{.}(2013)]%
        {gcloud}
\bibfield{author}{\bibinfo{person}{Sheng Di}, \bibinfo{person}{Derrick Kondo}, {and} \bibinfo{person}{Franck Cappello}.} \bibinfo{year}{2013}\natexlab{}.
\newblock \showarticletitle{Characterizing Cloud Applications on a Google Data Center}. In \bibinfo{booktitle}{\emph{Proceedings of the 2013 42nd International Conference on Parallel Processing}} \emph{(\bibinfo{series}{ICPP '13})}. \bibinfo{publisher}{IEEE Computer Society}, \bibinfo{address}{USA}, \bibinfo{pages}{468–473}.
\newblock
\showISBNx{9780769551173}
\href{https://doi.org/10.1109/ICPP.2013.56}{doi:\nolinkurl{10.1109/ICPP.2013.56}}


\bibitem[Dukkipati et~al\mbox{.}(2005)]%
        {rcp}
\bibfield{author}{\bibinfo{person}{Nandita Dukkipati}, \bibinfo{person}{Masayoshi Kobayashi}, \bibinfo{person}{Rui Zhang-Shen}, {and} \bibinfo{person}{Nick McKeown}.} \bibinfo{year}{2005}\natexlab{}.
\newblock \showarticletitle{Processor Sharing Flows in the Internet}. In \bibinfo{booktitle}{\emph{Quality of Service -- IWQoS 2005}}, \bibfield{editor}{\bibinfo{person}{Hermann de~Meer} {and} \bibinfo{person}{Nina Bhatti}} (Eds.). \bibinfo{publisher}{Springer Berlin Heidelberg}, \bibinfo{address}{Berlin, Heidelberg}, \bibinfo{pages}{271--285}.
\newblock
\showISBNx{978-3-540-31659-6}


\bibitem[Giannakou et~al\mbox{.}(2025)]%
        {hpc_qos}
\bibfield{author}{\bibinfo{person}{Anna Giannakou}, \bibinfo{person}{Jonathan Skone}, \bibinfo{person}{Vinay Sawal}, \bibinfo{person}{Ronal Kumar}, \bibinfo{person}{Stephen Simms}, \bibinfo{person}{Nicholas Wright}, {and} \bibinfo{person}{Lavanya Ramakrishnan}.} \bibinfo{year}{2025}\natexlab{}.
\newblock \showarticletitle{Implementing Network-level QoS at HPC Datacenters to Enable Distributed Scientific Workflows}. In \bibinfo{booktitle}{\emph{Proceedings of the SC '25 Workshops of the International Conference for High Performance Computing, Networking, Storage and Analysis}} \emph{(\bibinfo{series}{SC Workshops '25})}. \bibinfo{publisher}{Association for Computing Machinery}, \bibinfo{address}{New York, NY, USA}, \bibinfo{pages}{919–928}.
\newblock
\showISBNx{9798400718717}
\href{https://doi.org/10.1145/3731599.3767450}{doi:\nolinkurl{10.1145/3731599.3767450}}


\bibitem[Handley et~al\mbox{.}(2017)]%
        {ndp}
\bibfield{author}{\bibinfo{person}{Mark Handley}, \bibinfo{person}{Costin Raiciu}, \bibinfo{person}{Alexandru Agache}, \bibinfo{person}{Andrei Voinescu}, \bibinfo{person}{Andrew~W. Moore}, \bibinfo{person}{Gianni Antichi}, {and} \bibinfo{person}{Marcin W\'{o}jcik}.} \bibinfo{year}{2017}\natexlab{}.
\newblock \showarticletitle{Re-architecting datacenter networks and stacks for low latency and high performance}. In \bibinfo{booktitle}{\emph{Proceedings of the Conference of the ACM Special Interest Group on Data Communication}} (Los Angeles, CA, USA) \emph{(\bibinfo{series}{SIGCOMM '17})}. \bibinfo{publisher}{Association for Computing Machinery}, \bibinfo{address}{New York, NY, USA}, \bibinfo{pages}{29–42}.
\newblock
\showISBNx{9781450346535}
\href{https://doi.org/10.1145/3098822.3098825}{doi:\nolinkurl{10.1145/3098822.3098825}}


\bibitem[He et~al\mbox{.}(2015)]%
        {presto}
\bibfield{author}{\bibinfo{person}{Xin He}, \bibinfo{person}{Hongqiang~Harry Zhang}, \bibinfo{person}{Yuliang Li}, \bibinfo{person}{Weitao Wang}, \bibinfo{person}{Hongqiang~Harry Liu}, {and} \bibinfo{person}{Kai Chen}.} \bibinfo{year}{2015}\natexlab{}.
\newblock \showarticletitle{Presto: Edge-based Load Balancing for Fast Datacenter Networks}. In \bibinfo{booktitle}{\emph{Proceedings of the 2015 ACM Conference on SIGCOMM}} \emph{(\bibinfo{series}{SIGCOMM '15})}. \bibinfo{publisher}{Association for Computing Machinery}, \bibinfo{address}{New York, NY, USA}.
\newblock


\bibitem[Hong et~al\mbox{.}(2012)]%
        {premtive_schdle}
\bibfield{author}{\bibinfo{person}{Chi-Yao Hong}, \bibinfo{person}{Matthew Caesar}, {and} \bibinfo{person}{P.~Brighten Godfrey}.} \bibinfo{year}{2012}\natexlab{}.
\newblock \showarticletitle{Finishing flows quickly with preemptive scheduling}.
\newblock \bibinfo{journal}{\emph{SIGCOMM Comput. Commun. Rev.}} \bibinfo{volume}{42}, \bibinfo{number}{4} (\bibinfo{date}{aug} \bibinfo{year}{2012}), \bibinfo{pages}{127–138}.
\newblock
\showISSN{0146-4833}
\href{https://doi.org/10.1145/2377677.2377710}{doi:\nolinkurl{10.1145/2377677.2377710}}


\bibitem[Hong et~al\mbox{.}(2013)]%
        {swan}
\bibfield{author}{\bibinfo{person}{Chi-Yao Hong}, \bibinfo{person}{Srikanth Kandula}, \bibinfo{person}{Ratul Mahajan}, \bibinfo{person}{Ming Zhang}, \bibinfo{person}{Vijay Gill}, \bibinfo{person}{Mohan Nanduri}, {and} \bibinfo{person}{Roger Wattenhofer}.} \bibinfo{year}{2013}\natexlab{}.
\newblock \showarticletitle{Achieving high utilization with software-driven WAN}.
\newblock \bibinfo{journal}{\emph{SIGCOMM Comput. Commun. Rev.}} \bibinfo{volume}{43}, \bibinfo{number}{4} (\bibinfo{date}{Aug.} \bibinfo{year}{2013}), \bibinfo{pages}{15–26}.
\newblock
\showISSN{0146-4833}
\href{https://doi.org/10.1145/2534169.2486012}{doi:\nolinkurl{10.1145/2534169.2486012}}


\bibitem[Jacobson(1988)]%
        {vjacobson_cc}
\bibfield{author}{\bibinfo{person}{V. Jacobson}.} \bibinfo{year}{1988}\natexlab{}.
\newblock \showarticletitle{Congestion avoidance and control}.
\newblock \bibinfo{journal}{\emph{SIGCOMM Comput. Commun. Rev.}} \bibinfo{volume}{18}, \bibinfo{number}{4} (\bibinfo{date}{Aug.} \bibinfo{year}{1988}), \bibinfo{pages}{314–329}.
\newblock
\showISSN{0146-4833}
\href{https://doi.org/10.1145/52325.52356}{doi:\nolinkurl{10.1145/52325.52356}}


\bibitem[Katabi et~al\mbox{.}(2002)]%
        {xcp}
\bibfield{author}{\bibinfo{person}{Dina Katabi}, \bibinfo{person}{Mark Handley}, {and} \bibinfo{person}{Charlie Rohrs}.} \bibinfo{year}{2002}\natexlab{}.
\newblock \showarticletitle{Congestion control for high bandwidth-delay product networks}.
\newblock \bibinfo{journal}{\emph{SIGCOMM Comput. Commun. Rev.}} \bibinfo{volume}{32}, \bibinfo{number}{4} (\bibinfo{date}{aug} \bibinfo{year}{2002}), \bibinfo{pages}{89–102}.
\newblock
\showISSN{0146-4833}
\href{https://doi.org/10.1145/964725.633035}{doi:\nolinkurl{10.1145/964725.633035}}


\bibitem[Kelly(2003)]%
        {num-kelly}
\bibfield{author}{\bibinfo{person}{Frank Kelly}.} \bibinfo{year}{2003}\natexlab{}.
\newblock \showarticletitle{Fairness and Stability of End-to-End Congestion Control*}.
\newblock \bibinfo{journal}{\emph{European Journal of Control}} \bibinfo{volume}{9}, \bibinfo{number}{2} (\bibinfo{year}{2003}), \bibinfo{pages}{159--176}.
\newblock
\showISSN{0947-3580}
\href{https://doi.org/10.3166/ejc.9.159-176}{doi:\nolinkurl{10.3166/ejc.9.159-176}}


\bibitem[Kumar et~al\mbox{.}(2015)]%
        {bwe}
\bibfield{author}{\bibinfo{person}{Alok Kumar}, \bibinfo{person}{Sushant Jain}, \bibinfo{person}{Uday Naik}, \bibinfo{person}{Anand Raghuraman}, \bibinfo{person}{Nikhil Kasinadhuni}, \bibinfo{person}{Enrique~Cauich Zermeno}, \bibinfo{person}{C.~Stephen Gunn}, \bibinfo{person}{Jing Ai}, \bibinfo{person}{Bj\"{o}rn Carlin}, \bibinfo{person}{Mihai Amarandei-Stavila}, \bibinfo{person}{Mathieu Robin}, \bibinfo{person}{Aspi Siganporia}, \bibinfo{person}{Stephen Stuart}, {and} \bibinfo{person}{Amin Vahdat}.} \bibinfo{year}{2015}\natexlab{}.
\newblock \showarticletitle{BwE: Flexible, Hierarchical Bandwidth Allocation for WAN Distributed Computing}. In \bibinfo{booktitle}{\emph{Proceedings of the 2015 ACM Conference on Special Interest Group on Data Communication}} (London, United Kingdom) \emph{(\bibinfo{series}{SIGCOMM '15})}. \bibinfo{publisher}{Association for Computing Machinery}, \bibinfo{address}{New York, NY, USA}, \bibinfo{pages}{1–14}.
\newblock
\showISBNx{9781450335423}
\href{https://doi.org/10.1145/2785956.2787478}{doi:\nolinkurl{10.1145/2785956.2787478}}


\bibitem[Lazar and Semret(1999a)]%
        {psp_for_network}
\bibfield{author}{\bibinfo{person}{Aurel Lazar} {and} \bibinfo{person}{Nemo Semret}.} \bibinfo{year}{1999}\natexlab{a}.
\newblock \showarticletitle{Design and Analysis of the Progressive Second Price Auction for Network Bandwidth Sharing}.
\newblock  (\bibinfo{date}{11} \bibinfo{year}{1999}).
\newblock


\bibitem[Lazar and Semret(1999b)]%
        {psp2}
\bibfield{author}{\bibinfo{person}{Aurel Lazar} {and} \bibinfo{person}{Nemo Semret}.} \bibinfo{year}{1999}\natexlab{b}.
\newblock \showarticletitle{The Progressive Second Price Auction Mechanism for Network Resource Sharing}.
\newblock \bibinfo{journal}{\emph{International Symposium on Dynamic Games and Applications,}} (\bibinfo{date}{05} \bibinfo{year}{1999}).
\newblock


\bibitem[Li et~al\mbox{.}(2024)]%
        {qclimb}
\bibfield{author}{\bibinfo{person}{Wenxin Li}, \bibinfo{person}{Xin He}, \bibinfo{person}{Yuan Liu}, \bibinfo{person}{Keqiu Li}, \bibinfo{person}{Kai Chen}, \bibinfo{person}{Zhao Ge}, \bibinfo{person}{Zewei Guan}, \bibinfo{person}{Heng Qi}, \bibinfo{person}{Song Zhang}, {and} \bibinfo{person}{Guyue Liu}.} \bibinfo{year}{2024}\natexlab{}.
\newblock \showarticletitle{Flow Scheduling with Imprecise Knowledge}. In \bibinfo{booktitle}{\emph{21st USENIX Symposium on Networked Systems Design and Implementation (NSDI 24)}}. \bibinfo{publisher}{USENIX Association}, \bibinfo{address}{Santa Clara, CA}, \bibinfo{pages}{95--111}.
\newblock
\showISBNx{978-1-939133-39-7}
\urldef\tempurl%
\url{https://www.usenix.org/conference/nsdi24/presentation/li-wenxin}
\showURL{%
\tempurl}


\bibitem[Li et~al\mbox{.}(2019)]%
        {hpcc}
\bibfield{author}{\bibinfo{person}{Yuliang Li}, \bibinfo{person}{Rui Miao}, \bibinfo{person}{Hongqiang~Harry Liu}, \bibinfo{person}{Yan Zhuang}, \bibinfo{person}{Fei Feng}, \bibinfo{person}{Lingbo Tang}, \bibinfo{person}{Zheng Cao}, \bibinfo{person}{Ming Zhang}, \bibinfo{person}{Frank Kelly}, \bibinfo{person}{Mohammad Alizadeh}, {and} \bibinfo{person}{Minlan Yu}.} \bibinfo{year}{2019}\natexlab{}.
\newblock \showarticletitle{HPCC: high precision congestion control}. In \bibinfo{booktitle}{\emph{Proceedings of the ACM Special Interest Group on Data Communication}} (Beijing, China) \emph{(\bibinfo{series}{SIGCOMM '19})}. \bibinfo{publisher}{Association for Computing Machinery}, \bibinfo{address}{New York, NY, USA}, \bibinfo{pages}{44–58}.
\newblock
\showISBNx{9781450359566}
\href{https://doi.org/10.1145/3341302.3342085}{doi:\nolinkurl{10.1145/3341302.3342085}}


\bibitem[Maill{\'e} and Tuffin(2004)]%
        {psp_4}
\bibfield{author}{\bibinfo{person}{Patrick Maill{\'e}} {and} \bibinfo{person}{Bruno Tuffin}.} \bibinfo{year}{2004}\natexlab{}.
\newblock \showarticletitle{Multi-bid Versus Progressive Second Price Auctions in a Stochastic Environment}. In \bibinfo{booktitle}{\emph{Quality of Service in the Emerging Networking Panorama}}, \bibfield{editor}{\bibinfo{person}{Josep Sol{\'e}-Pareta}, \bibinfo{person}{Michael Smirnov}, \bibinfo{person}{Piet Van~Mieghem}, \bibinfo{person}{Jordi Domingo-Pascual}, \bibinfo{person}{Edmundo Monteiro}, \bibinfo{person}{Peter Reichl}, \bibinfo{person}{Burkhard Stiller}, {and} \bibinfo{person}{Richard~J. Gibbens}} (Eds.). \bibinfo{publisher}{Springer Berlin Heidelberg}, \bibinfo{address}{Berlin, Heidelberg}, \bibinfo{pages}{318--327}.
\newblock
\showISBNx{978-3-540-30193-6}


\bibitem[Maillé(2007)]%
        {psp_1}
\bibfield{author}{\bibinfo{person}{Patrick Maillé}.} \bibinfo{year}{2007}\natexlab{}.
\newblock \showarticletitle{Market Clearing Price and Equilibria of the Progressive Second Price Mechanism}.
\newblock \bibinfo{journal}{\emph{http://dx.doi.org/10.1051/ro:2007030}}  \bibinfo{volume}{41} (\bibinfo{date}{10} \bibinfo{year}{2007}).
\newblock
\href{https://doi.org/10.1051/ro:2007030}{doi:\nolinkurl{10.1051/ro:2007030}}


\bibitem[Montazeri et~al\mbox{.}(2018)]%
        {homa}
\bibfield{author}{\bibinfo{person}{Behnam Montazeri}, \bibinfo{person}{Yilong Li}, \bibinfo{person}{Mohammad Alizadeh}, {and} \bibinfo{person}{John Ousterhout}.} \bibinfo{year}{2018}\natexlab{}.
\newblock \showarticletitle{Homa: a receiver-driven low-latency transport protocol using network priorities}. In \bibinfo{booktitle}{\emph{Proceedings of the 2018 Conference of the ACM Special Interest Group on Data Communication}} (Budapest, Hungary) \emph{(\bibinfo{series}{SIGCOMM '18})}. \bibinfo{publisher}{Association for Computing Machinery}, \bibinfo{address}{New York, NY, USA}, \bibinfo{pages}{221–235}.
\newblock
\showISBNx{9781450355674}
\href{https://doi.org/10.1145/3230543.3230564}{doi:\nolinkurl{10.1145/3230543.3230564}}


\bibitem[Nathan et~al\mbox{.}(2019)]%
        {minerva}
\bibfield{author}{\bibinfo{person}{Vikram Nathan}, \bibinfo{person}{Vibhaalakshmi Sivaraman}, \bibinfo{person}{Ravichandra Addanki}, \bibinfo{person}{Mehrdad Khani}, \bibinfo{person}{Prateesh Goyal}, {and} \bibinfo{person}{Mohammad Alizadeh}.} \bibinfo{year}{2019}\natexlab{}.
\newblock \showarticletitle{End-to-end transport for video QoE fairness}. In \bibinfo{booktitle}{\emph{Proceedings of the ACM Special Interest Group on Data Communication}} (Beijing, China) \emph{(\bibinfo{series}{SIGCOMM '19})}. \bibinfo{publisher}{Association for Computing Machinery}, \bibinfo{address}{New York, NY, USA}, \bibinfo{pages}{408–423}.
\newblock
\showISBNx{9781450359566}
\href{https://doi.org/10.1145/3341302.3342077}{doi:\nolinkurl{10.1145/3341302.3342077}}


\bibitem[Ousterhout et~al\mbox{.}(2017)]%
        {flexplane}
\bibfield{author}{\bibinfo{person}{Amy Ousterhout}, \bibinfo{person}{Jonathan Perry}, \bibinfo{person}{Hari Balakrishnan}, {and} \bibinfo{person}{Petr Lapukhov}.} \bibinfo{year}{2017}\natexlab{}.
\newblock \showarticletitle{Flexplane: An Experimentation Platform for Resource Management in Datacenters}. In \bibinfo{booktitle}{\emph{14th USENIX Symposium on Networked Systems Design and Implementation (NSDI 17)}}. \bibinfo{publisher}{USENIX Association}, \bibinfo{address}{Boston, MA}, \bibinfo{pages}{438--451}.
\newblock
\showISBNx{978-1-931971-37-9}
\urldef\tempurl%
\url{https://www.usenix.org/conference/nsdi17/technical-sessions/presentation/ousterhout}
\showURL{%
\tempurl}


\bibitem[Perry et~al\mbox{.}(2014)]%
        {fastpass}
\bibfield{author}{\bibinfo{person}{Jonathan Perry}, \bibinfo{person}{Amy Ousterhout}, \bibinfo{person}{Hari Balakrishnan}, \bibinfo{person}{Devavrat Shah}, {and} \bibinfo{person}{Hans Fugal}.} \bibinfo{year}{2014}\natexlab{}.
\newblock \showarticletitle{Fastpass: a centralized "zero-queue" datacenter network}. In \bibinfo{booktitle}{\emph{Proceedings of the 2014 ACM Conference on SIGCOMM}} (Chicago, Illinois, USA) \emph{(\bibinfo{series}{SIGCOMM '14})}. \bibinfo{publisher}{Association for Computing Machinery}, \bibinfo{address}{New York, NY, USA}, \bibinfo{pages}{307–318}.
\newblock
\showISBNx{9781450328364}
\href{https://doi.org/10.1145/2619239.2626309}{doi:\nolinkurl{10.1145/2619239.2626309}}


\bibitem[Popa et~al\mbox{.}(2012)]%
        {faircloud}
\bibfield{author}{\bibinfo{person}{Lucian Popa}, \bibinfo{person}{Gautam Kumar}, \bibinfo{person}{Mosharaf Chowdhury}, \bibinfo{person}{Arvind Krishnamurthy}, \bibinfo{person}{Sylvia Ratnasamy}, {and} \bibinfo{person}{Ion Stoica}.} \bibinfo{year}{2012}\natexlab{}.
\newblock \showarticletitle{FairCloud: sharing the network in cloud computing}.
\newblock \bibinfo{journal}{\emph{SIGCOMM Comput. Commun. Rev.}} \bibinfo{volume}{42}, \bibinfo{number}{4} (\bibinfo{date}{Aug.} \bibinfo{year}{2012}), \bibinfo{pages}{187–198}.
\newblock
\showISSN{0146-4833}
\href{https://doi.org/10.1145/2377677.2377717}{doi:\nolinkurl{10.1145/2377677.2377717}}


\bibitem[Qureshi et~al\mbox{.}(2022)]%
        {plb}
\bibfield{author}{\bibinfo{person}{Mubashir~Adnan Qureshi}, \bibinfo{person}{Yuchung Cheng}, \bibinfo{person}{Qianwen Yin}, \bibinfo{person}{Qiaobin Fu}, \bibinfo{person}{Gautam Kumar}, \bibinfo{person}{Masoud Moshref}, \bibinfo{person}{Junhua Yan}, \bibinfo{person}{Van Jacobson}, \bibinfo{person}{David Wetherall}, {and} \bibinfo{person}{Abdul Kabbani}.} \bibinfo{year}{2022}\natexlab{}.
\newblock \showarticletitle{PLB: congestion signals are simple and effective for network load balancing}. In \bibinfo{booktitle}{\emph{Proceedings of the ACM SIGCOMM 2022 Conference}}. \bibinfo{pages}{207--218}.
\newblock


\bibitem[Radunovi\'{c} and Boudec(2007)]%
        {max_min}
\bibfield{author}{\bibinfo{person}{Bozidar Radunovi\'{c}} {and} \bibinfo{person}{Jean-Yves~Le Boudec}.} \bibinfo{year}{2007}\natexlab{}.
\newblock \showarticletitle{A unified framework for max-min and min-max fairness with applications}.
\newblock \bibinfo{journal}{\emph{IEEE/ACM Trans. Netw.}} \bibinfo{volume}{15}, \bibinfo{number}{5} (\bibinfo{date}{Oct.} \bibinfo{year}{2007}), \bibinfo{pages}{1073–1083}.
\newblock
\showISSN{1063-6692}
\href{https://doi.org/10.1109/TNET.2007.896231}{doi:\nolinkurl{10.1109/TNET.2007.896231}}


\bibitem[Rajasekaran et~al\mbox{.}(2024a)]%
        {cassini_nsdi24}
\bibfield{author}{\bibinfo{person}{Sudarsanan Rajasekaran}, \bibinfo{person}{Manya Ghobadi}, {and} \bibinfo{person}{Aditya Akella}.} \bibinfo{year}{2024}\natexlab{a}.
\newblock \showarticletitle{{CASSINI}: {Network-Aware} Job Scheduling in Machine Learning Clusters}. In \bibinfo{booktitle}{\emph{21st USENIX Symposium on Networked Systems Design and Implementation (NSDI 24)}}. \bibinfo{publisher}{USENIX Association}, \bibinfo{address}{Santa Clara, CA}, \bibinfo{pages}{1403--1420}.
\newblock
\showISBNx{978-1-939133-39-7}
\urldef\tempurl%
\url{https://www.usenix.org/conference/nsdi24/presentation/rajasekaran}
\showURL{%
\tempurl}


\bibitem[Rajasekaran et~al\mbox{.}(2022)]%
        {cassini_hotnets}
\bibfield{author}{\bibinfo{person}{Sudarsanan Rajasekaran}, \bibinfo{person}{Manya Ghobadi}, \bibinfo{person}{Gautam Kumar}, {and} \bibinfo{person}{Aditya Akella}.} \bibinfo{year}{2022}\natexlab{}.
\newblock \showarticletitle{{Congestion Control in Machine Learning Clusters}}. In \bibinfo{booktitle}{\emph{Proceedings of the 21st ACM Workshop on Hot Topics in Networks}} (Austin, Texas) \emph{(\bibinfo{series}{HotNets '22})}. \bibinfo{pages}{235–242}.
\newblock


\bibitem[Rajasekaran et~al\mbox{.}(2024b)]%
        {mltcp_hotnets}
\bibfield{author}{\bibinfo{person}{Sudarsanan Rajasekaran}, \bibinfo{person}{Sanjoli Narang}, \bibinfo{person}{Anton~A. Zabreyko}, {and} \bibinfo{person}{Manya Ghobadi}.} \bibinfo{year}{2024}\natexlab{b}.
\newblock \showarticletitle{MLTCP: A Distributed Technique to Approximate Centralized Flow Scheduling For Machine Learning}. In \bibinfo{booktitle}{\emph{Proceedings of the 23rd ACM Workshop on Hot Topics in Networks}} (Irvine, CA, USA) \emph{(\bibinfo{series}{HotNets '24})}. \bibinfo{publisher}{Association for Computing Machinery}, \bibinfo{address}{New York, NY, USA}, \bibinfo{pages}{167–176}.
\newblock
\showISBNx{9798400712722}
\href{https://doi.org/10.1145/3696348.3696878}{doi:\nolinkurl{10.1145/3696348.3696878}}


\bibitem[Rajasekaran et~al\mbox{.}(2024c)]%
        {mltcp_arxiv}
\bibfield{author}{\bibinfo{person}{Sudarsanan Rajasekaran}, \bibinfo{person}{Sanjoli Narang}, \bibinfo{person}{Anton~A. Zabreyko}, {and} \bibinfo{person}{Manya Ghobadi}.} \bibinfo{year}{2024}\natexlab{c}.
\newblock \bibinfo{title}{MLTCP: Congestion Control for DNN Training}.
\newblock
\showeprint[arxiv]{2402.09589}~[cs.NI]
\urldef\tempurl%
\url{https://arxiv.org/abs/2402.09589}
\showURL{%
\tempurl}


\bibitem[Roy et~al\mbox{.}(2015)]%
        {fb_dc}
\bibfield{author}{\bibinfo{person}{Arjun Roy}, \bibinfo{person}{Hongyi Zeng}, \bibinfo{person}{Jasmeet Bagga}, \bibinfo{person}{George Porter}, {and} \bibinfo{person}{Alex~C. Snoeren}.} \bibinfo{year}{2015}\natexlab{}.
\newblock \showarticletitle{Inside the Social Network's (Datacenter) Network}.
\newblock \bibinfo{journal}{\emph{SIGCOMM Comput. Commun. Rev.}} \bibinfo{volume}{45}, \bibinfo{number}{4} (\bibinfo{date}{Aug.} \bibinfo{year}{2015}), \bibinfo{pages}{123–137}.
\newblock
\showISSN{0146-4833}
\href{https://doi.org/10.1145/2829988.2787472}{doi:\nolinkurl{10.1145/2829988.2787472}}


\bibitem[Schrage(1968)]%
        {SRPT_optimality_1968}
\bibfield{author}{\bibinfo{person}{Linus Schrage}.} \bibinfo{year}{1968}\natexlab{}.
\newblock \showarticletitle{Letter to the Editor—A Proof of the Optimality of the Shortest Remaining Processing Time Discipline}.
\newblock \bibinfo{journal}{\emph{Operations Research}} \bibinfo{volume}{16}, \bibinfo{number}{3} (\bibinfo{year}{1968}), \bibinfo{pages}{687--690}.
\newblock
\href{https://doi.org/10.1287/opre.16.3.687}{doi:\nolinkurl{10.1287/opre.16.3.687}}


\bibitem[Schrage and Miller(1966)]%
        {SRPT_1966}
\bibfield{author}{\bibinfo{person}{Linus~E. Schrage} {and} \bibinfo{person}{Louis~W. Miller}.} \bibinfo{year}{1966}\natexlab{}.
\newblock \showarticletitle{The Queue M/G/1 with the Shortest Remaining Processing Time Discipline}.
\newblock \bibinfo{journal}{\emph{Operations Research}} \bibinfo{volume}{14}, \bibinfo{number}{4} (\bibinfo{year}{1966}), \bibinfo{pages}{670--684}.
\newblock
\href{https://doi.org/10.1287/opre.14.4.670}{doi:\nolinkurl{10.1287/opre.14.4.670}}


\bibitem[Shenker et~al\mbox{.}(1997)]%
        {rfc_qos}
\bibfield{author}{\bibinfo{person}{S. Shenker}, \bibinfo{person}{C. Partridge}, {and} \bibinfo{person}{R. Guerin}.} \bibinfo{year}{1997}\natexlab{}.
\newblock \bibinfo{title}{RFC2212: Specification of Guaranteed Quality of Service}.
\newblock


\bibitem[Tuffin(2002)]%
        {psp_2}
\bibfield{author}{\bibinfo{person}{Bruno Tuffin}.} \bibinfo{year}{2002}\natexlab{}.
\newblock \showarticletitle{Revisited Progressive Second Price Auction for Charging Telecommunication Networks}.
\newblock \bibinfo{journal}{\emph{Telecommunication Systems}} \bibinfo{volume}{20}, \bibinfo{number}{3} (\bibinfo{date}{07} \bibinfo{year}{2002}), \bibinfo{pages}{255--263}.
\newblock
\showISSN{1572-9451}
\href{https://doi.org/10.1023/A:1016545228543}{doi:\nolinkurl{10.1023/A:1016545228543}}


\bibitem[Vamanan et~al\mbox{.}(2012)]%
        {d2tcp}
\bibfield{author}{\bibinfo{person}{Balajee Vamanan}, \bibinfo{person}{Jahangir Hasan}, {and} \bibinfo{person}{T.N. Vijaykumar}.} \bibinfo{year}{2012}\natexlab{}.
\newblock \showarticletitle{Deadline-aware datacenter tcp (D2TCP)}. In \bibinfo{booktitle}{\emph{Proceedings of the ACM SIGCOMM 2012 Conference on Applications, Technologies, Architectures, and Protocols for Computer Communication}} (Helsinki, Finland) \emph{(\bibinfo{series}{SIGCOMM '12})}. \bibinfo{publisher}{Association for Computing Machinery}, \bibinfo{address}{New York, NY, USA}, \bibinfo{pages}{115–126}.
\newblock
\showISBNx{9781450314190}
\href{https://doi.org/10.1145/2342356.2342388}{doi:\nolinkurl{10.1145/2342356.2342388}}


\bibitem[Wellman et~al\mbox{.}(2001)]%
        {mop}
\bibfield{author}{\bibinfo{person}{Michael~P. Wellman}, \bibinfo{person}{William~E. Walsh}, \bibinfo{person}{Peter~R. Wurman}, {and} \bibinfo{person}{Jeffrey~K. MacKie-Mason}.} \bibinfo{year}{2001}\natexlab{}.
\newblock \showarticletitle{Auction Protocols for Decentralized Scheduling}.
\newblock \bibinfo{journal}{\emph{Games and Economic Behavior}} \bibinfo{volume}{35}, \bibinfo{number}{1} (\bibinfo{year}{2001}), \bibinfo{pages}{271--303}.
\newblock
\showISSN{0899-8256}
\href{https://doi.org/10.1006/game.2000.0822}{doi:\nolinkurl{10.1006/game.2000.0822}}


\bibitem[Wilson et~al\mbox{.}(2011)]%
        {d3}
\bibfield{author}{\bibinfo{person}{Christo Wilson}, \bibinfo{person}{Hitesh Ballani}, \bibinfo{person}{Thomas Karagiannis}, {and} \bibinfo{person}{Ant Rowtron}.} \bibinfo{year}{2011}\natexlab{}.
\newblock \showarticletitle{Better never than late: meeting deadlines in datacenter networks}. In \bibinfo{booktitle}{\emph{Proceedings of the ACM SIGCOMM 2011 Conference}} (Toronto, Ontario, Canada) \emph{(\bibinfo{series}{SIGCOMM '11})}. \bibinfo{publisher}{Association for Computing Machinery}, \bibinfo{address}{New York, NY, USA}, \bibinfo{pages}{50–61}.
\newblock
\showISBNx{9781450307970}
\href{https://doi.org/10.1145/2018436.2018443}{doi:\nolinkurl{10.1145/2018436.2018443}}


\bibitem[Xia et~al\mbox{.}(2021)]%
        {mocc_ieee}
\bibfield{author}{\bibinfo{person}{Zhenchang Xia}, \bibinfo{person}{Yanjiao Chen}, \bibinfo{person}{Libing Wu}, \bibinfo{person}{Yu-Cheng Chou}, \bibinfo{person}{Zhicong Zheng}, \bibinfo{person}{Haoyang Li}, {and} \bibinfo{person}{Baochun Li}.} \bibinfo{year}{2021}\natexlab{}.
\newblock \showarticletitle{A Multi-objective Reinforcement Learning Perspective on Internet Congestion Control}. In \bibinfo{booktitle}{\emph{2021 IEEE/ACM 29th International Symposium on Quality of Service (IWQOS)}}. \bibinfo{pages}{1--10}.
\newblock
\href{https://doi.org/10.1109/IWQOS52092.2021.9521291}{doi:\nolinkurl{10.1109/IWQOS52092.2021.9521291}}


\bibitem[Zhang et~al\mbox{.}(2022)]%
        {aqua}
\bibfield{author}{\bibinfo{person}{Yiwen Zhang}, \bibinfo{person}{Gautam Kumar}, \bibinfo{person}{Nandita Dukkipati}, \bibinfo{person}{Xian Wu}, \bibinfo{person}{Priyaranjan Jha}, \bibinfo{person}{Mosharaf Chowdhury}, {and} \bibinfo{person}{Amin Vahdat}.} \bibinfo{year}{2022}\natexlab{}.
\newblock \showarticletitle{Aequitas: admission control for performance-critical RPCs in datacenters}. In \bibinfo{booktitle}{\emph{Proceedings of the ACM SIGCOMM 2022 Conference}} (Amsterdam, Netherlands) \emph{(\bibinfo{series}{SIGCOMM '22})}. \bibinfo{publisher}{Association for Computing Machinery}, \bibinfo{address}{New York, NY, USA}, \bibinfo{pages}{1–18}.
\newblock
\showISBNx{9781450394208}
\href{https://doi.org/10.1145/3544216.3544271}{doi:\nolinkurl{10.1145/3544216.3544271}}


\bibitem[Zhang et~al\mbox{.}(2025)]%
        {virtual_qos}
\bibfield{author}{\bibinfo{person}{Zhaochen Zhang}, \bibinfo{person}{Feiyang Xue}, \bibinfo{person}{Keqiang He}, \bibinfo{person}{Zhimeng Yin}, \bibinfo{person}{Gianni Antichi}, \bibinfo{person}{Jiaqi Gao}, \bibinfo{person}{Yizhi Wang}, \bibinfo{person}{Rui Ning}, \bibinfo{person}{Haixin Nan}, \bibinfo{person}{Xu Zhang}, \bibinfo{person}{Peirui Cao}, \bibinfo{person}{Xiaoliang Wang}, \bibinfo{person}{Wanchun Dou}, \bibinfo{person}{Guihai Chen}, {and} \bibinfo{person}{Chen Tian}.} \bibinfo{year}{2025}\natexlab{}.
\newblock \showarticletitle{Enabling Virtual Priority in Data Center Congestion Control}. In \bibinfo{booktitle}{\emph{Proceedings of the Twentieth European Conference on Computer Systems}} (Rotterdam, Netherlands) \emph{(\bibinfo{series}{EuroSys '25})}. \bibinfo{publisher}{Association for Computing Machinery}, \bibinfo{address}{New York, NY, USA}, \bibinfo{pages}{396–412}.
\newblock
\showISBNx{9798400711961}
\href{https://doi.org/10.1145/3689031.3717463}{doi:\nolinkurl{10.1145/3689031.3717463}}


\bibitem[Zhu et~al\mbox{.}(2015)]%
        {dcqcn}
\bibfield{author}{\bibinfo{person}{Yibo Zhu}, \bibinfo{person}{Haggai Eran}, \bibinfo{person}{Daniel Firestone}, \bibinfo{person}{Chuanxiong Guo}, \bibinfo{person}{Marina Lipshteyn}, \bibinfo{person}{Yehonatan Liron}, \bibinfo{person}{Jitendra Padhye}, \bibinfo{person}{Shachar Raindel}, \bibinfo{person}{Mohamad~Haj Yahia}, {and} \bibinfo{person}{Ming Zhang}.} \bibinfo{year}{2015}\natexlab{}.
\newblock \showarticletitle{Congestion Control for Large-Scale RDMA Deployments}. In \bibinfo{booktitle}{\emph{Proceedings of the 2015 ACM Conference on Special Interest Group on Data Communication}} (London, United Kingdom) \emph{(\bibinfo{series}{SIGCOMM '15})}. \bibinfo{publisher}{Association for Computing Machinery}, \bibinfo{address}{New York, NY, USA}, \bibinfo{pages}{523–536}.
\newblock
\showISBNx{9781450335423}
\href{https://doi.org/10.1145/2785956.2787484}{doi:\nolinkurl{10.1145/2785956.2787484}}


\end{thebibliography}
\end{small}
\appendix
\clearpage
\nobalance

\appendix
\section{Truthfulness in Repeated Second-Price Auctions}
\label{sec:app_ic}
\begin{theorem}[Myopic Truthfulness under Price-Taking Behaviour]
If the market is price-taking, then bidding myopic marginal utility $b_i(h_t)=\Delta_0 (h_t)$ in each round $t$ is a weakly dominant strategy.
\end{theorem}
\begin{proof}
Consider a bidder $i$ participating in repeated sealed-bid second-price auctions until it wins $N_i$ rounds. Let $k$ denote the current auction round and let bidder $i$ submit bid
\[
b_i^k=\pi(h_k),
\]
where $h_k$ denotes the bidder's observed history up to round $k$, consisting only of internal state information (excluding any bid-dependent feedback). For Markov policies, $h_k$ reduces to the state $S_k$. Let
\[
M_k=\max_{j\neq i} b_j^k
\]
be the highest competing bid (i.e., the price paid upon winning).

\medskip
\noindent
\textbf{Exogenous price assumption.}
We assume prices are exogenous in the sense that deviating in round $k$ does not affect the distribution of future clearing prices $\{M_t\}_{t>k}$ (conditioned on the same observed history).

\medskip
\noindent

\textbf{Expected continuation utilities and notation.}
Fix any baseline bidding policy $\pi$ for bidder $i$.
At each round $k$, the bid submitted under $\pi$ is $b_i^k=\pi(h_k)$.
For any policy $\pi'$, let $W_{k+1}^{\pi'}$ and $L_{k+1}^{\pi'}$ denote bidder $i$'s expected \emph{net} utility from rounds $(k\!+\!1)$ onward if bidder $i$ wins and loses the current ($k^{th}$) round, respectively, conditioned on the observed history $h_k$ and assuming it continues with $\pi'$ thereafter.
(For clarity: these continuation utilities already incorporate all future payments under the exogenous price process.)
Note that $h_{k+1}$ differs depending on whether bidder $i$ wins the current auction or not even for the same prior history $h_k$.

Let $r_w(h_k)$ and $r_l(h_k)$ denote the immediate valuation obtained in the current round under win and loss, respectively.
Define the one-step marginal willingness to pay at round $k$ under continuation policy $\pi'$ as
\begin{equation}
\label{eq:mar_contin_clean}
\Delta_{\pi'}(h_k)
=
\big(W_{k+1}^{\pi'}+r_w(h_k)\big)
-
\big(L_{k+1}^{\pi'}+r_l(h_k)\big).
\end{equation}
In particular, under the baseline policy $\pi$ we write $\Delta_{\pi}(h_k)$.

\medskip
\noindent
\textbf{Net utility from round $k$.}
Let $J_i^{\pi}(h_k)$ denote bidder $i$'s expected net utility from round $k$ onward under policy $\pi$. Conditioning on the outcome of the current auction yields
\begin{equation}
\label{eq:J_clean}
\begin{split}
J_i^{\pi}(h_k)
&=
\Big(W_{k+1}^{\pi}+r_w(h_k)\Big)
\Pr(\pi(h_k)>M_k\mid h_k)\\
&\quad+
\Big(L_{k+1}^{\pi}+r_l(h_k)\Big)
\Pr(\pi(h_k)\le M_k\mid h_k)\\
&\quad-
\mathbb{E}\!\left[M_k\,\mathbb{I}\{\pi(h_k)>M_k\}\mid h_k\right].
\end{split}
\end{equation}

Substituting \eqref{eq:mar_contin_clean} into \eqref{eq:J_clean} gives
\begin{equation}
\label{eq:J_reduced}
J_i^{\pi}(h_k)
=
\mathbb{E}\!\left[(\Delta_{\pi}(h_k)-M_k)
\mathbb{I}\{\pi(h_k)>M_k\}\mid h_k\right]
+C(h_k),
\end{equation}
where $C(h_k)$ collects terms independent of the current bid $\pi(h_k)$.
\begin{equation}
\begin{split}
    & C(h_k) = L_{k+1}^{\pi} + r_l(h_k)
\end{split}
\end{equation}
Both $L_{k+1}^{\pi}$ and $r_l(h_k)$ are determined by $h_k$ and the event of losing the current auction, and do not depend on the current bid.

\medskip
\noindent
\textbf{Single-round deviation.}
Now define policy $G_{\pi}^k$ that coincides with $\pi$ everywhere except at round $k$, where it bids $\Delta_{\pi}(h_k)$.
Under exogenous prices, the distribution of future clearing prices $\{M_t\}_{t>k}$ is unaffected by this $k^{th}$ round bid deviation, implying that $C(h_k)$ remains unchanged given history $h_k$.
Thus the only difference between $J_i^{G_{\pi}^k}$ and $J_i^{\pi}$ is the outcome at round $k$, yielding
\begin{equation}
\label{eq:comp}
\begin{split}
 & \implies J_i^{G_{\pi}^k}(k) - J_i^{\pi}(k) = \\
 & \mathbb{E} \left[ (\Delta_{\pi}(h_k) - M_k) \left[ \mathbb{I}(\Delta_{\pi}(h_k) > M_k) - \mathbb{I}(\pi(h_k) > M_k) \right] \mid h_k \right]
\end{split}
\end{equation}
From equation~\ref{eq:comp}, we have $J_i^{G_{\pi}^k}(h_k)-J_i^{\pi}(h_k)\ge 0$.

Since the policies $\pi$ and $G_{\pi}^k$ coincide in all rounds prior to $k$, they induce the same realized history and accumulate identical utility up to round $k$.
Therefore, the only possible difference in net utility arises from their behavior at round $k$ and beyond. Evaluating utility from the beginning of the auction sequence,
\begin{equation}
J_i^{G_{\pi}^k}-J_i^{\pi}
=
J_i^{G_{\pi}^k}(h_k)-J_i^{\pi}(h_k)
\ge 0.
\end{equation}

Hence, replacing $\pi$ by $G_{\pi}^k$ weakly improves bidder $i$'s net utility.
Equivalently, viewing $\pi$ as the baseline policy $G_{\pi}^{-1}$, we have shown that for any round $k$,
\begin{equation}
\label{eq:mainresult_clean}
J_i^{G_{\pi}^k} \;\ge\; J_i^{G_{\pi}^{-1}},
\qquad \forall k\ge 0.
\end{equation}

\textbf{Recursive application.}
We now apply this argument recursively across rounds.
Start with the policy $G_{\pi}^0$, which coincides with $\pi$ except that it bids $\Delta_{\pi}(h_0)$ at round $0$.
Next, define $G_{G_{\pi}^0}^1$ to coincide with $G_{\pi}^0$ except that it also bids $\Delta_{G_{\pi}^0}(h_1)$ at round $1$, and so on.
The subscript in $\Delta_{\pi'}(\cdot)$ tracks the continuation policy used to define the marginal continuation utility; after we modify earlier rounds, the continuation policy (and hence the appropriate $\Delta_{\pi'}$) changes accordingly.
By \eqref{eq:mainresult_clean}, each step weakly improves utility.

Let $\Delta_0$ denote the limiting policy obtained by applying this replacement at every round, i.e., the policy that bids $b_i^k=\Delta_0(h_k)$ where $\Delta_0(h_k)$ is the marginal continuation utility induced by the (fully updated) policy itself.
Then $\Delta_0$ weakly dominates $\pi$.
\end{proof}

\newpage
\section{Bidding API}
\label{sec:app_api}
Using the objective-utility model and dynamic-programming method in Section~\S\ref{sec:rl}, we summarize bidding policies for common objectives. We write $\mathcal{F}$ for the CDF of the (global) second-price distribution. Let $S_t$ denote the bidder state at round $t$, $b(S_t)$ its bid, and $P_t$ the highest competing bid (i.e., the second price) drawn from $\mathcal{F}$. Winning occurs when $P_t < b(S_t)$, so $\Pr(\text{win})=\mathcal{F}(b(S_t))$.
\\~\\

\textbf{(i) Flow Completion Time (FCT).}
\begin{itemize}
    \item State variable $S_t$ is remaining FCT in units of RTT.
    \item Objective value $o$ is equal to $S_t$.
    \item Example utility: $u(o) = C - w\times (o - \frac{o^2}{2T})$, where $w,T$ determine the budget/penalty scale.
    \item Derivative utility: $\frac{du}{do} = -w\times (1 - \frac{o}{T})$.
    \item Immediate rewards: $r_w(S_t) = 0$ and $r_l(S_t) = -w\times (1 - \frac{S_t}{T})$ (a loss increases remaining FCT by one RTT).
\end{itemize}
The state-value recursion is:
\begin{equation}
\begin{split}
U(S_t)
&= \Big(U(S_t - 1) - \mathbb{E}[P_t\mid P_t < b(S_t)]\Big) \times \Pr(P_t < b(S_t))\\
&\quad + \Big(U(S_t) + r_l(S_t)\Big) \times \Pr(P_t \ge b(S_t)).
\end{split}
\end{equation}

Equivalently,
\begin{equation}
\label{eq:srpt_u}
\begin{split}
U(S_t)
&= \Big(U(S_t - 1) - \mathbb{E}[P_t\mid P_t < b(S_t)]\Big) \times \mathcal{F}(b(S_t))\\
&\quad + \Big(U(S_t) + r_l(S_t)\Big) \times \big(1 - \mathcal{F}(b(S_t))\big).
\end{split}
\end{equation}

The bid is the marginal continuation utility:
\begin{equation}
    b(S_t) = \Delta U(S_t) = U(S_t - 1) - U(S_t) - r_l(S_t).
\end{equation}

This admits a closed-form characterization (no dynamic programming needed):
\begin{equation}
    \label{eq:srpt_bid}
    \boxed{\textcolor{blue}{\int_{0}^{b(S_t)} \mathcal{F}(x)\,dx = r_l(S_t)}}
\end{equation}

Equivalently, maximizing $U(S_t)$ with respect to $b(S_t)$ in Equation~\ref{eq:srpt_u} yields the same condition, verifying that bidding the marginal utility is a dominant strategy in the induced per-round second-price auction.
\\~\\

\textbf{(ii) Deadline.}
\begin{itemize}
    \item State variable $S_t$ is (remaining FCT, deadline slack), written as $(F,D)$ in units of RTT.
    \item Objective value is non-zero only upon meeting the deadline: reward $C$.
    \item Immediate rewards: $r_w(F,D)=0$ and $r_l(F,D)=0$ everywhere except at $F=0$, where $r_w(0,D)=C$.
\end{itemize}
The state-value recursion is:
\begin{equation}
\label{eq:dl_u}
\begin{split}
U(F,D)
&= \Big(U(F-1, D) + r_w(F,D)\Big) \times \mathcal{F}(b(S_t))\\
&\quad + U(F-1, D-1) \times \big(1 - \mathcal{F}(b(S_t))\big).
\end{split}
\end{equation}
The marginal continuation utility is:
\begin{equation}
    b(F,D) = \Delta U(F,D) = U(F-1,D) - U(F-1, D-1) + r_w(F,D).
\end{equation}
The above equations are solvable using 2-D dynamic programming.
\\~\\

\textbf{(iii) Fairness.}
\begin{itemize}
    \item State variable $S_t$ is EWMA-averaged bandwidth (EWMA factor $\alpha$).
    \item Objective value $o$ is equal to $S_t$.
    \item Example utility: $u(o) = w\times \log(o)$, where $w$ determines the budget scale.
    \item Derivative utility: $\frac{du}{do} = \frac{w}{o}$.
    \item Upon win, $o$ increases by $\alpha(1-o)$; upon loss it decreases by $\alpha o$.
    \item Immediate rewards: $r_w(S_t) = \frac{w}{S_t} \times \alpha (1 - S_t)$ and $r_l(S_t) = - w \times \alpha $.
\end{itemize}
The state-value recursion is:
\begin{equation}
\label{eq:fr_u}
\begin{split}
U(S_t)
&= \Big(U((1 - \alpha)S_t + \alpha) + r_w(S_t)\Big) \times \mathcal{F}(b(S_t))\\
&\quad + \Big(U((1 - \alpha)S_t) + r_l(S_t)\Big) \times \big(1 - \mathcal{F}(b(S_t))\big).
\end{split}
\end{equation}
The marginal utility is:
\begin{equation}
    b(S_t) = \Delta U(S_t) = U((1 - \alpha)S_t + \alpha) - U((1 - \alpha)S_t) + r_w(S_t) - r_l(S_t).
\end{equation}
The above equations are solvable using iterative dynamic programming.
\\~\\
\textbf{(iv) Coflow Completion Time (CCT).}
We use the FCT bidding agent discussed in (i) but with the change that remaining size is replaced with remaining coflow bytes on the bottleneck link. It requires a per coflow-centric agent that knows the sizes and src-dst pairs of all flows within that coflow. This effectively mimics SEBF scheduling policy used in Varys~\cite{varys}. For full Varys style bidding, we need MADD which can be emulated using our deadline bidding API. Bottleneck flows should use SEBF bidding while non-bottleneck flows should use deadline bidding with bottleneck FCT as their deadline. 
\\~\\
\textbf{(v) Best Effort (Fixed Service Tier).}
This is a default bidding solution for participants only having the notion of entitlement and not of urgency and are not particularly interested in any objectives. It simply bids a constant value per RTT.

\newpage
\section{Convergence and Stability}
\label{sec:app_conv}
We stress the market with abrupt changes and quantify convergence using the \emph{average market price} (the mean of the global price distribution $\mathcal{F}$) measured once per epoch. Note that each epoch is an independent simulation having thousands of flows/bidding agents.

\para{Cold-start/Warm-up.}
When the market is in infancy with no bidding agents, we start from an assumed uniform prior over the global price signal $\mathcal{F}$ and allow bidding agents to update their policies using it. Then the iterative process of observing the real $\mathcal{F}$ and updating bidding policies continues until a fixed point is reached where $\mathcal{F}$ stops changing. Figure~\ref{fig:conv_init} shows the evolution of the average price with (a) direct updates using $\mathcal{F}$ and (b) smoothed updates using EWMA on $\mathcal{F}$ with the older distribution. In both cases, prices converge within roughly 4--6 epochs, while EWMA smoothing reduces transient overshoot. The reducing gap between average prices over epochs verifies the contractive behaviour expected of the market operator, leading to quick convergence.

\para{Step Change in Load.}
Next, we increase offered load from 40\% to 80\% in a single epoch and observe how quickly the prices adjust. Figure~\ref{fig:conv_load} shows that the average price rises to the new operating point within roughly 4--6 epochs, again with smaller oscillations under EWMA smoothing, showing the same contractive behaviour.

\begin{figure}[h]
    \centering
    \includegraphics[width=0.45\textwidth]{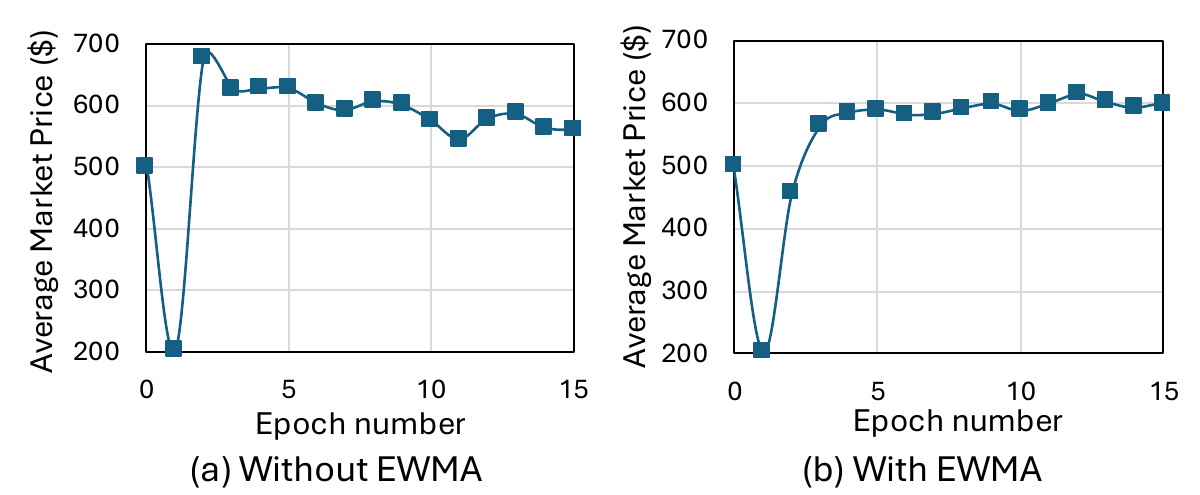}
    \caption{Convergence under cold-start of market.}
    \label{fig:conv_init}
\end{figure}

\begin{figure}[h]
    \centering
    \includegraphics[width=0.45\textwidth]{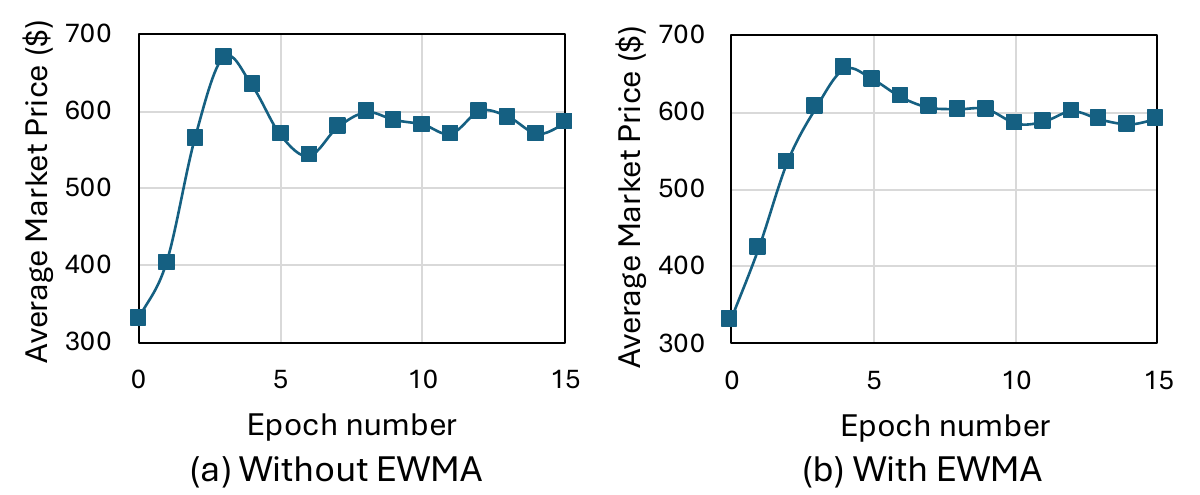}
    \caption{Convergence under step change in load.}
    \label{fig:conv_load}
\end{figure}
\label{totalpages}
\end{document}